\documentclass[a4paper,12pt,openany]{article}

\usepackage[polutonikogreek,latin,english]{babel}

\usepackage[english]{babel}
\usepackage{amsmath}
\usepackage{amsthm}
\usepackage{enumerate}
\usepackage{graphicx}
\usepackage{mathrsfs}
\usepackage{amssymb}
\usepackage{color}
\usepackage{empheq}
\usepackage{rotating}
\usepackage{empheq}
\usepackage{braket}
\usepackage{fancyhdr}
\usepackage{hyperref}
\hypersetup{
colorlinks = true, 
linkcolor = {blue}
}

\usepackage[affil-it]{authblk}

\numberwithin{equation}{section}

\newtheorem{prop}{Proposition}[section]
\newtheorem{lemma}{Lemma}[section]
\newtheorem{thm}{Theorem}[section]
\newtheorem{oss}{Remark}[section]

\theoremstyle{definition}

\begin{document}
\title{Riemann-Hilbert approach to gap probabilities for the Bessel process}
\author{Manuela Girotti
\thanks {Electronic address: \texttt{mgirotti@mathstat.concordia.ca}}}
\affil{\small Department of Mathematics and Statistics, Concordia University \\
1455 de Maisonneuve Ouest, Montr\'eal, Qu\'ebec, Canada, H3G 1M8}
\date{}

\maketitle




\begin{abstract}
We consider the gap probability for the Bessel process in the single-time and multi-time case. We prove that the scalar and matrix Fredholm determinants of such process can be expressed in terms of  determinants of integrable kernels \'a la Its-Izergin-Korepin-Slavnov and thus related to suitable Riemann-Hilbert problems. In the single-time case, we construct a Lax pair formalism and we derive a Painlev\'e III equation related to the Fredholm determinant.
\end{abstract}

\section{Introduction}

The Bessel process is a determinantal point  process \cite{Soshnikov} defined in terms of a trace-class integral operator acting on $L^2(\mathbb{R}_+)$, with the
kernel
\begin{equation}
K_B(x,y) = \frac{J_\nu(\sqrt{x})\sqrt{y}J_{\nu+1}(\sqrt{y}) - J_{\nu+1}(\sqrt{x})\sqrt{x}J_{\nu}(\sqrt{y})}{2(x-y)} \label{officialBessel}
\end{equation}
where $J_\nu$ are Bessel functions with parameter $\nu >-1$. 

The Bessel kernel $K_B$ arose originally as the correlation function in the hard edge scaling limit of the Laguerre and Jacobi Unitary Ensembles (\cite{forresterLUE}, \cite{NagFor}, \cite{NagWad}) as well as of generalized LUEs and JUEs (\cite{KuijVan}, \cite{justVan}). 

In this article we onsider the gap probabilities of this process, namely, the probabilites that no particles fall within a determined Borel subset of the real axis. In particular, we will be concerned with the Fredholm determinant of such operator on a collection of (finite) intervals $I:= \bigcup_{i=1}^{N-1} [a_i, a_{i+1}]$, i.e. the Tracy-Widom distribution $\det \left( \text{Id} -   K_B\chi_{I} \right)$, $\chi_{I}$ being the characteristic function of $I$,
and the emphasis is on the determinant thought of as function of the endpoint $a_i$, $i=1,\ldots, N$.

Such Fredholm determinant  appears in Random Matrix Theory and can be interpreted as the gap probability for Laguerre and Jacobi Unitary ensembles in the hard edge scaling limit. We recall that the gap probability is the probability that no eigenvalue of the random matrix lies in a given Borel set $I$ (see \cite{Mehta}, \cite{TWrandom}, \cite{TWFredh} for the general theory).  

We refer to \cite{TWBessel} for gap probabilities of the Bessel process. Nevertheless, we point out that the equations found in the present paper are not the same as those shown in \cite{TWBessel} and they are derived thorough a completely different method.

The second part of this paper will examine the Bessel process in a time-dependent regime. 

Dyson (\cite{Dyson}) described how to implement a dynamics into random matrix models in such a way that the eigenvalues of the matrix behave like finitely many non-intersecting Brownian motions on the real line. 

The idea of the introduction of multi-time processes is thus to generalize the notion of matrix ensemble so that the model of non-intersecting Brownian paths acquires a meaning, not only as a static model in timeless thermodynamical equilibrium, but as a dynamical system which may be in an arbitrary non-equilibrium state changing with time.

Given times $0<\tau_1<\ldots<\tau_n<T$ and subsets $I_k\subset \mathbb{R}$, $k=1,\ldots,n$, the quantity of interest is the probability that for all $k$ no curve passes through $I_k$ at time $\tau_k$, i.e. no eigenvalue lies in $I_k$ at time $\tau_k$ for all $k$. We call again this quantity ``gap probability".

It can be shown (using well-known results from \cite{EM} and  \cite{KarlinMcGregor}) that the gap probability is again equal to the Fredholm determinant of a suitable (multi-time) kernel.

The multi-time Bessel kernel (see \cite{Japan} and \cite{TWMBessel}) is a matrix kernel with entries 
\begin{equation}
K_{B; ij} (x,y)= \left\{ \begin{array}{ll}
\int_0^1 e^{u \Delta} J_{\nu}(\sqrt{xu})J_{\nu}(\sqrt{yu})\, du & i \geq j \\
- \int_1^\infty e^{u \Delta} J_{\nu}(\sqrt{xu})J_{\nu}(\sqrt{yu})\, du & i < j 
\end{array}
\right.
\end{equation}
with $\Delta := \Delta_{ij} = \tau_i - \tau_j$, $\nu>-1$.

\begin{oss}
In the case $T = \tau_1 = \ldots = \tau_n =0$, we can recover the Bessel kernel (\ref{officialBessel}).
\end{oss}

As shown by Forrester, Nagao and Honner in \cite{forrester}, the multi-time Bessel process (with its correspondent kernel) appears as scaling limit of the Extended Laguerre process at the hard edge of the spectrum. 

\paragraph{Fredholm determinant of integrable kernels and $\tau$-function.}
The Fredholm determinant of the time-less Bessel kernel and, as it will be clear in the paper, the Fredholm determinant of its multi-time counterpart 
are instances of determinants of operators with integrable kernel in the sense of Its-Izergin-Korepin-Slavnov \cite{IIKS}. We will recall here the main results that will be used in the paper (for a concise and clear description we refer to \cite{JohnSasha}). 

Let consider a $p\times p$ matrix operator $K$ acting on $L^2(\Sigma)$, with $\Sigma \subseteq \mathbb{C}$ a collection of piecewise smooth, oriented contours (possibly extending to infinity), defined as
\begin{equation*}
K(\phi)(\lambda) = \int_{\Sigma} K(\lambda, \mu)\phi(\mu) \, d\mu \label{JSkernel}
\end{equation*}
with kernel 
\begin{equation*}
K(\lambda, \mu) = \frac{\textbf{f}^T(\lambda) \textbf{g} (\mu)}{\lambda - \mu} \label{JSkernel2}
\end{equation*}
where $\textbf{f}, \textbf{g}$ are rectangular $r \times p$ matrix valued functions ($p < r$), such that $\textbf{f} ^T(\lambda) \textbf{g}(\lambda) =0$. We assume that $\textbf{f}$ and $\textbf{g}$ are sufficiently smooth functions along the connected components of $\Sigma$.

We are interested in its Fredholm determinant $\det (\text{Id} - K)$. The key observation is the Jacobi's formula for the variation of the determinant
\begin{equation}
\partial \, \ln \det (\text{Id}-K) = -  \text{Tr} \left( (\text{Id} +R) \, \partial K \right) \label{tauFredholm}
\end{equation}
where $\partial$ signifies the differential with respect to any auxiliary parameters (let's denote them by $\vec s$), on which $K$ may depend, and $R$ is the resolvent operator $R := (I - K)^{-1} K$. Moreover, $R$ belongs to the same class of operators as well
\begin{equation*}
R(\lambda, \mu) := \frac{\textbf{F}^T(\lambda) \textbf{G}(\mu)}{\lambda - \mu}
\end{equation*}
where $\textbf{F}$ and $\textbf{G}$ can be determined by solving the following Riemann-Hilbert problem (RHP)
\begin{subequations}
\begin{align}
&\Gamma_+(\lambda) = \Gamma_- (\lambda) M( \lambda )\ \ \ \lambda \in \Sigma \label{jump} \\
&\Gamma(\lambda) = \textbf{I}_r + \mathcal{O}\left( \frac{1}{\lambda} \right) \ \ \ \lambda \rightarrow \infty  \label{inftyexp} \\
& M(\lambda) = \textbf{I}_r - 2\pi i \textbf{f}(\lambda) \textbf{g}^T(\lambda) \\
&\textbf{F}(\lambda) = \Gamma(\lambda) \textbf{f}(\lambda), \ \ \ \ \ \textbf{G}(\lambda) = \Gamma(\lambda) \textbf{g}(\lambda)
\end{align}
\end{subequations}
(the dependence on the parameters $\vec s$ is implicit).

In several cases of interest, and in our case as well, the Fredholm determinant for such a kernel coincides with the notion of the isomonodromic $\tau$-function  introduced by Jimbo, Miwa and Ueno (\cite{JMUII}, \cite{JMUIII} and \cite{JMU}) to study monodromy-preserving deformations of rational connections on $\mathbb{P}^1$. 

We report here just the main results that will be used later. We refer to \cite{Misomonodromic} and \cite{MeM} for a thorough exposition.

First of all we consider a slightly more general notion of $\tau$-function associated to any Riemann-Hilbert problem depending on parameters.

Given a generic RHP (\ref{jump})-(\ref{inftyexp}), we introduce the following one form on the space of (deformation) parameters 
\begin{subequations}
\begin{align}
&\omega (\partial) := \int_\Sigma \operatorname{Tr}\left( \Gamma_-^{-1}(\lambda) \Gamma_-'(\lambda) \Xi_\partial (\lambda) \right) \, \frac{d\lambda}{2\pi i} \\
&\Xi_\partial(\lambda) := \partial M (\lambda) M^{-1}(\lambda);
\end{align}
\end{subequations}
for the sake of simplicity, we are assuming that all the expressions involved are depending smoothly on the parameters (here and below we will denote with $'$ the derivative with respect to $\lambda$).

Then, the following equality holds
\begin{equation}
\omega (\partial) = \partial \ln \det (I-K) - H(M) \label{Misomonodromictau}
\end{equation}
where 
\begin{equation*}
H(M) := \int_{\Sigma} \left( \partial \textbf{f}\,  '\, ^T \textbf{g} + \textbf{f}\, '\, ^T \partial \textbf{g} \right) d\lambda - 2\pi i \int_{\Sigma} \textbf{g}^T \textbf{f}\, ' \partial \textbf{g}^T \textbf{f} \, d\lambda. \label{Hannoying}
\end{equation*}

On the other hand, it is possible to show that $\omega_M$ is the logarithmic total differential of the isomonodromic $\tau$-function of  Jimbo-Miwa-Ueno, in the case of RHPs that correspond to rational ODEs.

Thus, it is possible to define, up to normalization, the isomonodromic $\tau$-function  $\tau_{JMU}: = \text{exp}\left(\int \omega \right)$ which, thanks to (\ref{Misomonodromictau}),  will coincide with the Fredholm determinant $\det(I - K)$, if $H(M) \equiv 0$.

\bigskip

The main steps in our study of the gap probabilities for the Bessel process are the following:  we will first find an integrable operator in the sense specified above, acting on $L^2(\Sigma)$, with $\Sigma$ a suitable collection of contours. Through an appropriate Fourier transform, we will prove that such operator has the same Fredholm determinant as the Bessel process. We will then set up a RHP for this integrable operator and connect it to the Jimbo-Miwa-Ueno $\tau$ function. 

This strategy will be applied separately to both the single-time and the multi-time Bessel process. Our approach derives from the one used in \cite{MeMmulti} and \cite{MeM} for the Airy and Pearcey processes in the dynamic and time-less regime respectively.

Whereas the part dedicated to the single-time process is mostly a review of known results (see \cite{Ptrans}, \cite{JMUII} and \cite{TWBessel}), re-derived using an alternative approach, the results on the multi-time Bessel are new and never appeared in the literature before.

The present paper is organized as follows: in section \ref{singletimeBesselPIII} we will deal with the single-time Bessel determinant field in the general case of several intervals; in the subsection \ref{N1trans} we will focus on the single-time Bessel process restricted to a single interval $[0,a]$: we will find a Lax pair and we will be able to make a connection between the Fredholm determinant and the third Painlev\'e transcendent. This provides a different and direct proof of this known connection (\cite{JMUII}, \cite{TWBessel}); in particular our approach directly specifies the monodromy data of the associated isomonodromic system and allows to use the steepest descent method to investigate asymptotic properties, if so desired. In section \ref{gapprobmultiB} we will study the gap probabilities for the multi-time Bessel process. Although the results of section \ref{gapprobmultiB} strictly include those of  section \ref{singletimeBesselPIII}, we have decided to separate the two cases for the benefit of a clearer exposition.

\section{Gap probabilities for the single-time Bessel process and Painlev\'e Transcendent}
\label{singletimeBesselPIII}

We recall the definition of the Bessel kernel 
\begin{equation}
K_B(x,y) =  \left(\frac{y}{x}\right)^{\nu/2} \iint_{\gamma \times \hat \gamma} \frac{e^{xt - \frac{1}{4t} -ys + \frac{1}{4s}}}{t-s} \left(\frac{s}{t}\right)^\nu   \frac{dt}{2\pi i}\, \frac{ds}{2\pi i}  \label{Bker}
\end{equation}
with $\nu>-1$, $x,y>0$ and $\gamma$ a curve that extends to $- \infty$ and winds around the zero counterclockwise, while the curve $\hat \gamma$ is simply the transformed curve under the map $t \rightarrow 1/t$; the logarithmic cut is on $\mathbb{R}_{-}$.
The contours are as in figure \ref{Besselcurvepic}.
\begin{figure}[!h]
\centering 
\includegraphics[width = .7 \textwidth]{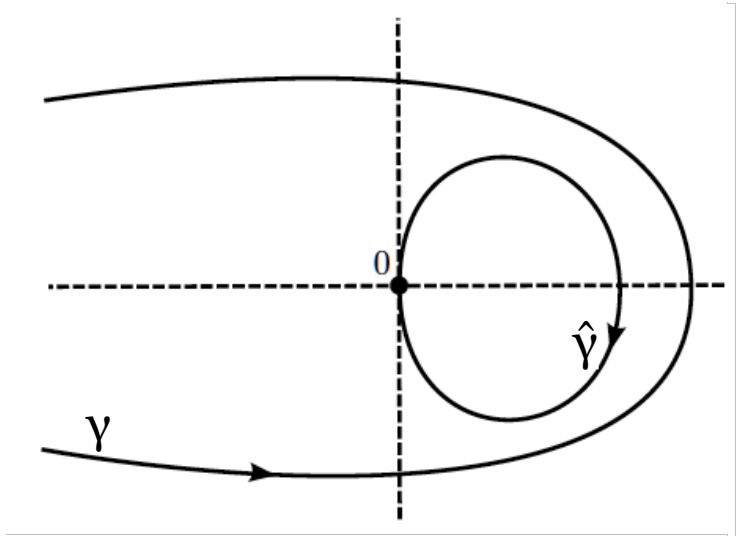}
\caption{The contours appearing in the definition of the Bessel kernel. 
}
\label{Besselcurvepic}
\end{figure}

We want to study the Fredholm determinant of the Bessel operator; in particular, we will focus on the following quantity 
\begin{equation}
 \det \left(\text{Id} - K_B \chi_{I}\right)
\end{equation}
where $\chi_I$ is the characteristic function of a collection of finite intervals $I:= [a_1,a_2] \cup [a_3, a_4] \cup \ldots \cup [a_{N-1}, a_{N}]$, which assumes values $0$ or $1$ depending on the variable being respectively outside or within $I$.

\begin{oss}
The Bessel operator is not trace-class on an infinite interval. Thus, it is meaningless to consider the operator restricted to an infinite interval. 
\end{oss}

\begin{oss}
Defining $K_a := K_B(x,y) \chi_{[0,a]}(y)$, then we have
\begin{equation}
K_B (x,y) \cdot \chi_{I}(y) := \sum_{j=1}^N (-1)^j K_{a_j}(x,y)
\end{equation}
\end{oss}

Our goal is to setup a Riemann-Hilbert problem associated to the Fredholm determinant of $K_B\chi_I$. 

\begin{thm}\label{SB1}
The following identity between Fredholm determinants holds
\begin{equation}
\det \left(\operatorname{Id} - K_B\chi_I \right) = \det \left( \operatorname{Id} - \mathbb{B} \right)
\end{equation}
where $\mathbb{B}$ is a trace-class integrable operator 
acting on $L^2(\gamma \cup \hat \gamma)$ with kernel
\begin{subequations}
\begin{align}
&\mathbb{B}(s,t):= \frac{\vec{f}(s)^T \cdot \vec{g}(t)}{s-t} \label{Bintegrableker} \\
&\vec{f}(s) = \frac{1}{2\pi i} \left[ \begin{array}{c} e^{\frac{a_1s}{2}-\frac{1}{4s}}s^{-\nu} \\ 0 \\ 0 \\ \vdots \\ 0   \end{array} \right]\, \chi_{\gamma} (s) + \frac{1}{2\pi i} \left[ \begin{array}{c} 0 \\  e^{-a_1 s + \frac{1}{4s}}s^{\nu} \\ - e^{-a_2 s + \frac{1}{4s}}s^{\nu} \\ \vdots \\ (-1)^{N} e^{-a_{N-1}s + \frac{1}{4s}}s^{\nu} \\ (-1)^{N+1} e^{-a_Ns + \frac{1}{4s}}s^{\nu} \end{array} \right] \, \chi_{\hat \gamma}(s) \label{Bintegrableker3} \\
&\vec{g}(t) = \left[ \begin{array}{c} 0 \\ e^{\frac{a_1t}{2}} \\ e^{t\left(a_2-\frac{a_1}{2}\right)} \\ \vdots \\ e^{t\left( a_N-\frac{a_1}{2}\right)}  \end{array} \right] \, \chi_{\gamma}(t) + \left[ \begin{array}{c} 1 \\0 \\ 0 \\ \vdots \\0  \end{array} \right] \, \chi_{\hat \gamma}(t) \label{Bintegrableker2}
\end{align} 
\end{subequations}
\end{thm}

\begin{proof}
We work on a single kernel $K_{a_j}$ and we will later sum them up, thanks to the linearity of the operations that we are going to perform.

First of all, we can notice that, if $x<0$ or $y<0$, $K_B(x,y)\equiv 0$; then, using Cauchy's theorem, we can write
\begin{gather}
 K_{a_j}(x,y)= K_B(x,y)  \chi_{[0,a_j]}(y) = \nonumber \\
 =   \int_{i\mathbb{R}+\epsilon} \frac{d\xi}{2\pi i} \, \frac{e^{\xi (a_j-y)}}{\xi - s} \iint_{\hat \gamma \times \gamma} \frac{e^{ xt - \frac{1}{4t} -a_js + \frac{1}{4s}}}{t-s} \left(\frac{s}{t}\right)^\nu   \frac{dt\, ds}{(2\pi i)^2} = \nonumber \\
 =  \int_{i\mathbb{R}+\epsilon} \frac{d\xi}{2\pi i} \, e^{- \xi y} \int_{i\mathbb{R} + \epsilon} \frac{dt}{2\pi i}\, e^{xt} 
 \int_{\hat \gamma} \frac{ds}{2\pi i} \, \frac{e^{ \xi a_j - \frac{1}{4t} -a_j s + \frac{1}{4s}}}{(\xi - s)(t-s)} \left(\frac{s}{t}\right)^\nu 
\end{gather}
where $i\mathbb{R} + \epsilon$ ($\epsilon>0$) is a translated imaginary axis and we continuously deformed the curve $\gamma$ into such translated imaginary axis, in order to make the Fourier operator defined below more explicit. This is possible thanks to the analyticity of the kernel. We also discarded the conjugation term  $\left(\frac{y}{x}\right)^{\nu/2}$, due to the invariance of the Fredholm determinant under conjugation by a positive function.

Defining the following Fourier transform: 
\begin{equation}
\renewcommand\arraystretch{2}
\begin{array}{c|c}
\mathcal{F}: L^2(\mathbb{R}) \rightarrow L^2(i\mathbb{R}+\epsilon) & \mathcal{F}^{-1}: L^2(i\mathbb{R}+\epsilon) \rightarrow L^2(\mathbb{R}) \\
f(x) \mapsto \frac{1}{\sqrt{2\pi i}} \int_\mathbb{R}  f(x) e^{\xi x} dx & h(\xi) \mapsto \frac{1}{\sqrt{2\pi i}} \int_{i\mathbb{R}+\epsilon} h(\xi) e^{-\xi x}d\xi
\end{array}\label{Fourier}
\end{equation} 
it is straightforward to deduce that 
\begin{equation}
K_B \chi_I = \mathcal{F}^{-1} \circ \mathcal{K}_B \circ \mathcal{F}
\end{equation}
with $\mathcal{K}_B = \sum_j (-1)^j \mathcal{K}_{a_j}$ and $\forall \, j=1,\ldots, N$ $\mathcal{K}_{a_j}$ is an operator on $L^2(i\mathbb{R} +\epsilon)$ with kernel
\begin{equation*}
\mathcal{K}_{a_j} (\xi, t) =  \int_{\hat \gamma} \frac{ds}{2\pi i} \,  \frac{e^{ \xi a_j - \frac{1}{4t} -a_j s + \frac{1}{4s}}}{(\xi - s)(t-s)} \left(\frac{s}{t}\right)^\nu
\end{equation*}

In order to ensure convergence of the Fourier-transformed Bessel kernel, we conjugate $\mathcal{K}_B$  by a suitable function
\begin{align}
e^{\frac{a_1t}{2} - \frac{a_1\xi}{2}} \mathcal{K}_B (\xi, t) &=  \sum_{j=1}^N (-1)^j  \int_{\hat \gamma} \frac{ds}{2\pi i} \,  \frac{e^{ \xi \left(a_j-\frac{a_1}{2}\right) + \frac{a_1 t}{2}  - \frac{1}{4t} -a_j s + \frac{1}{4s}}}{(\xi - s)(t-s)} \left(\frac{s}{t}\right)^\nu \nonumber \\
&=: \sum_{j=1}^N (-1)^j \mathcal{K}_{a_j} 
\end{align}
and we continuously deform the translated imaginary axis $i\mathbb{R} + \epsilon$ into its original shape $\gamma$; note that $a_j-\frac{a_1}{2} >0$, $\forall \, j=1, \dots, N$.
With abuse of notation, we call the new kernels $\mathcal{K}_B$ and $\mathcal{K}_{a_j}$ as well.

\begin{lemma}\label{lemma21}
For each $j=1,\ldots, N$, $\mathcal{K}_{a_j}$ is trace-class. Moreover, the following decomposition holds $\mathcal{K}_{a_j} = \mathcal{A} \circ \mathcal{B}_{a_j}$, with
\begin{equation}
\renewcommand\arraystretch{2}
\begin{array}{c|c}
\mathcal{A}:  L^2(\hat \gamma) \rightarrow  L^2(\gamma) & \mathcal{B}_{a_j}: L^2(\gamma) \rightarrow L^2(\hat \gamma) \\
h(s) \mapsto t^{-\nu}e^{\frac{a_1t}{2} - \frac{1}{4t}}  \int_{\hat \gamma} \frac{h(s)}{t-s} \, \frac{ds}{2\pi i} & f(t) \mapsto s^\nu e^{-a_js + \frac{1}{4s}} \int_{\gamma} \frac{e^{t\left( a_j-\frac{a_1}{2} \right)}}{t -s} f(t)  \, \frac{dt}{2\pi i}
\end{array}
\end{equation}
\end{lemma}
\begin{proof}
It is easy to verify that $\mathcal{A}$ and $\mathcal{B}_{a_j}$ are Hilbert-Schmidt.

Moreover, we have the following decomposition of kernels. Introducing an additional contour $i\mathbb{R}+\delta$ not intersecting either of $\gamma, \hat \gamma$, we have $\mathcal{A} = \mathcal{P}_2 \circ \mathcal{P}_1$ with
\begin{equation}
\renewcommand\arraystretch{2}
\begin{array}{c|c}
\mathcal{P}_1: L^2(\hat \gamma) \rightarrow L^2(i\mathbb{R} + \delta) & \mathcal{P}_2: L^2(i\mathbb{R} + \delta) \rightarrow L^2(\gamma) \nonumber \\
\displaystyle \mathcal{P}_1[f] (u) = \int_{\hat \gamma} \frac{f(s)}{u-s} \frac{ds}{2\pi i}& \displaystyle \mathcal{P}_2[h] (t) = \frac{e^{\frac{a_1t}{2} - \frac{1}{4t}}}{ t^{\nu}} \int_{i\mathbb{R} + \delta} \frac{h(u)}{t-u} \frac{du}{2\pi i}
\end{array}
\end{equation}

Analogously, $\mathcal{B}_{a_j} = \mathcal{O}_{2,j} \circ \mathcal{O}_{1,j}$ with 
\begin{equation}
\renewcommand\arraystretch{2}
\begin{array}{c|c}
\mathcal{O}_{1,j}: L^2(\gamma) \rightarrow L^2(i\mathbb{R} + \delta) & \mathcal{O}_{2,j}: L^2(i\mathbb{R} + \delta) \rightarrow L^2(\hat \gamma) \nonumber \\
\displaystyle \mathcal{O}_{1,j}[f] (w) = \int_{\gamma}  \frac{e^{t\left( a_j -\frac{a_1}{2}\right)} f(t)}{t -w} \frac{dt}{2\pi i}& \displaystyle \mathcal{O}_{2,j}[h] (s) = s^\nu e^{-a_j s + \frac{1}{4s}}\int_{i\mathbb{R} + \delta} \frac{h(w)}{w-s} \frac{ds}{2\pi i}
\end{array}
\end{equation}
It is straightforward to check that $\mathcal{P}_i$ and $\mathcal{O}_{i,j}$ are Hilbert-Schmidt operators, $i=1,2$ and $j=1,\ldots, N$. Therefore, $\mathcal{A}$ and $\mathcal{B}_{a_j}$ are trace-class. 
\end{proof}

\begin{oss}
The kernel $\mathcal{A}$ does not depend on the set of parameters $\{ a_j\}_2^N$, but only on the first endpoint $a_1$.
\end{oss}

Before proceeding further, we notice that any operator acting on the Hilbert space $H := L^2(\gamma\cup \hat \gamma)\simeq L^2(\gamma) \oplus L^2(\hat \gamma) = H_1 \oplus H_2 $ can be written as a $2\times 2$ matrix of operators with $(i,j)$-entry given by an operator $H_i \rightarrow H_j$.

According to such split and using matrix notation, we can thus write $\det(\text{Id} - \mathcal{K}_B)$ as
\begin{gather}
 \det \left( \text{Id}_{L^2(\gamma)} -  \sum_{j=1}^N (-1)^j \mathcal{A} \circ \mathcal{B}_{a_j} \right) \nonumber \\ 
= \det \left( \text{Id}_{L^2(\gamma)} \otimes \text{Id}_{L^2(\hat \gamma)} -  \left[ \begin{array}{c|c}
0 & \mathcal{A} \\ \hline
\sum_{j=1}^N (-1)^j \mathcal{B}_{a_j}&0
\end{array} \right] \right)
= \det(\text{Id}_{L^2(\gamma \cup \hat \gamma)} - \mathbb{B})
\end{gather}
The first identity comes from multiplying the right hand side on the left by the following matrix (with determinant equal $1$)
\begin{equation*}
\text{Id}_{L^2(\gamma) \oplus L^2(\hat \gamma)} + \left[ \begin{array}{c|c} 
0 & -\mathcal{A} \\ \hline 
0 & 0  
\end{array} \right].
\end{equation*}
\end{proof}

\paragraph{The Riemann-Hilbert problem for the Bessel process.}
Thanks to Theorem \ref{SB1} we can relate the computation of the Fredholm determinat of the Bessel operator to the theory of isomonodromic equations. We start by setting up a suitable Riemann-Hilbert problem which
is naturally related to the Fredholm determinant of the operator $\mathbb{B}$ (see Figure \ref{Bessel1timejumps}).

\begin{prop}\label{propRHP}
Given the integrable kernel (\ref{Bintegrableker})-(\ref{Bintegrableker2}), the associated Riemann-Hilbert problem is the following:  
\begin{equation}
\left\{ \begin{array}{ll}
\Gamma_+(\lambda) = \Gamma_-(\lambda) \left(I - J(\lambda)\right)  & \lambda \in \Sigma:= \gamma \cup \hat \gamma \\
\Gamma(\lambda) = I + \mathcal{O}\left(\frac{1}{\lambda}\right) & \lambda \rightarrow \infty
\end{array}
\right. 
\label{BesselRHP}
\end{equation}
where $\Gamma$ is a $(N+1) \times (N+1)$ matrix such that it is analytic on $\mathbb{C} \backslash \Sigma$ and satisfies the jump conditions above with
\begin{gather}
J(\lambda) := 
 \left[ \begin{array}{ccccc}
0  & e^{\theta_1}& e^{\theta_2} & \ldots & e^{\theta_N} \\
0  &0& 0 &  \ldots & 0  \\
\vdots &   & &  & \vdots\\ 
0  & 0& 0&  \ldots & 0  \\
0  & 0& 0& \ldots & 0 
\end{array} \right] \chi_{\gamma}(\lambda) 
+ \left[ \begin{array}{cccc} 
0 & 0 & \ldots  & 0\\
e^{-\theta_1} & 0 & \ldots & 0 \\
- e^{-\theta_2} & 0 & \ldots & 0 \\
\vdots  & & & \vdots \\
(-1)^{N+1} e^{-\theta_N} & 0 & \ldots & 0
\end{array} \right] \chi_{\hat \gamma}(\lambda)
\end{gather}
$\theta_j := a_j\lambda - \frac{1}{4\lambda} - \nu \ln \lambda$, $\forall \, j =1,\ldots, N$.
\end{prop}
\begin{proof}
It is straightforward to verify that $I - J(\lambda) = I - \vec{f}(\lambda) \cdot \vec{g}(\lambda)^T$.
\end{proof}

\begin{thm}\label{JMUtau}
The Fredholm determinant $\det (\operatorname{Id}- K_B\chi_I)$ is equal to the isomonodromic $\tau$-function related to the Riemann-Hilbert problem defined in proposition \ref{propRHP}. 
In particular, $\forall \, j=1,\ldots ,N$
\begin{subequations}
\begin{gather}
\partial_{a_j} \ln \det \left(\operatorname{Id} - K_B\cdot \chi_{I} \right) = \int_\Sigma \operatorname{Tr} \left( \Gamma^{-1}_{-}(\lambda) \Gamma'_{-}(\lambda) \Xi_{\partial_{a_j}}(\lambda) \right) \, \frac{d\lambda}{2\pi i} \\
\Xi_{\partial}(\lambda) := - \partial J (\lambda) \cdot \left(I-J(\lambda) \right)^{-1}
\end{gather}
\end{subequations}
where $I= [a_1, a_2] \cup \ldots [a_{N-1}, a_N]$ is a collection of finite intervals and $\Sigma = \gamma \cup \hat \gamma$; we denote by $'$ the derivative with respect to the spectral parameter $\lambda$.
\end{thm}

\begin{proof}
Looking at the formula \ref{Misomonodromictau}, we just need to verify that $H(M)=0$, with $M(\lambda) := I - J(\lambda)$. 
\end{proof}

Starting from Theorem \ref{JMUtau}, it it possible to derive more explicit differential identities by using the Jimbo-Miwa-Ueno residue formula. First, we notice that the jump matrix $J(\lambda)$ can be written as
\begin{equation}
J(\lambda, \vec{a}) = e^{T(\lambda, \vec{a})} \cdot J_0 \cdot E^{-T(\lambda, \vec{a})}
\end{equation}
where $J_0$ is a constant matrix, consisting only on $0$ and $\pm1$, and
\begin{gather}
T(\lambda, \vec a) = \text{diag} \left(T_0, T_1, \ldots, T_N \right) \nonumber \\
T_0 = \frac{1}{N+1} \sum_{j=1}^N \theta_{j} \ \ \ \  T_j = T_0 - \theta_{j} 
\end{gather}
Therefore, the matrix $\Psi (\lambda, \vec a) := \Gamma (\lambda, \vec a) e^{T(\lambda, \vec a)}$ solves a Riemann-Hilbert problem with constant jumps and it is (sectionally) a solution to a polynomial ODE. 

Moreover, the following equality holds (see \cite[Theorem 2.1]{MeM})
\begin{equation}
 \int_\Sigma \operatorname{Tr} \left( \Gamma^{-1}_{-}(\lambda) \Gamma'_{-}(\lambda) \Xi_{\partial_{a_j}}(\lambda) \right) \, \frac{d\lambda}{2\pi i} 
= - \underset{\lambda = \infty}{\operatorname{res}} \operatorname{Tr}\left( \Gamma^{-1}(\lambda) \Gamma'(\lambda) \partial_{a_j}T(\lambda)    \right). \label{residue}
\end{equation}

In conclusion,
\begin{prop} \label{gamma122}
For all $j=1, \ldots, N$, the Fredholm determinant satisfies
\begin{equation}
\partial_{a_j} \ln \det \left( \operatorname{Id} - K_B \chi_{I} \right) = - \Gamma_{1; j+1, j+1}
\end{equation}
with $\Gamma_{1;j+1,j+1}$ the $(j+1, j+1)$ component of the residue matrix \,$\Gamma_1 = \lim_{\lambda \rightarrow \infty} \lambda \left(I - \Gamma(\lambda) \right) $ and $\Gamma$ is the solution to the Riemann-Hilbert problem (\ref{BesselRHP}).
\end{prop}

\begin{proof}(we refer to \cite{MeM})

Given the definition of $T(\lambda)$, 
\begin{equation}
 \partial_{a_j}T(\lambda, \vec a) =  \lambda \left( \frac{1}{N+1}I - E_{j+1, j+1} \right)
\end{equation}
and plugging into (\ref{residue}), we have
\begin{equation}
\partial_{a_j} \ln \det \left( \text{Id} - K_B \chi_{I} \right) = \frac{\operatorname{Tr} \Gamma_1}{N+1}  - \Gamma_{1; j+1,j+1}=  -\Gamma_{1; j+1, j+1}
\end{equation}
since $\det \Gamma (\lambda) \equiv 1$, thus $\operatorname{Tr}\Gamma_1 =0$.
\end{proof}

\subsection{The case $N=1$ and Painlev\'e III equation} \label{N1trans}

We consider now the case $N=1$, that is the case in which we study the Bessel kernel restricted to a single finite interval $[0,a]$. 

We will see that from the $2\times 2$ Bessel Riemann-Hilbert problem we will derive a suitable Lax pair which matches with the Lax pair of the Painlev\'e III transcendent, as shown in \cite{Ptrans} (the Lax pair described in \cite{Ptrans} is slightly different from the one found in our present thesis, but it can be shown that the two formulations are equivalent).

In order to make the connection with the Painlev\'e transcendent more explicit, we will work on a rescaled version of the Bessel kernel, which can be easily derived from our original definition (\ref{Bker}) through suitable scalings.

By specializing the results of the previous section, we get a (Fourier transformed) Bessel operator on $L^2(\gamma)$ with the following kernel
\begin{equation}
\mathcal{K}_B(\xi,t) = \int_{\hat \gamma} \frac{ds}{2\pi i} \,  \frac{e^{ \frac{\xi x}{4}  + \frac{x}{2} \left( \frac{t}{2}- \frac{1}{t} \right)- \frac{x}{2} \left( s - \frac{1}{s} \right) }}{(\xi - s)(t-s)} \left(\frac{s}{t}\right)^\nu
\end{equation}
where $x:= \sqrt{a}$. 

It can be easily shown that $\mathcal{K}_B$ is a trace-class operator, since product of two Hilbert-Schmidt operators $\mathcal{K}_B = \mathcal{A}_2 \circ \mathcal{A}_1$ with kernels
\begin{align}
\mathcal{A}_1 (t,s) &=\frac{1}{2\pi i} \frac{\text{exp}\left\{\frac{t x}{4} -\frac{x}{2} \left(s - \frac{1}{s} \right)\right\}}{t - s} \, s^\nu \cdot \chi_{\gamma}(t)\chi_{\hat \gamma}(s) \\
\mathcal{A}_2 (s,t) &= - \frac{1}{2\pi i} \frac{\text{exp}\left\{\frac{x}{2} \left( \frac{s}{2} - \frac{1}{s} \right)\right\}}{t-s} \, s^{-\nu} \cdot \chi_{\gamma}(s)\chi_{\hat \gamma}(t).
\end{align}

\begin{prop}
The operators $\mathcal{A}_j$, $j=1,2$, are trace-class.
\end{prop}
\begin{proof}
The proof follows the same arguments as the proof of Lemma \ref{lemma21}. 
\end{proof}

\begin{thm}
Consider the interval $[0,a]$, then the following identity holds
\begin{equation}
\det \left( \operatorname{Id} - K_B \chi_{[0,a]}\right) = \det \left( \operatorname{Id} - \mathbb{B} \right)
\end{equation}
with $\mathbb{B}$ a trace-class integrable operator with kernel defined as follows
\begin{subequations}
\begin{gather}
\mathbb{B}(t,s) = \frac{1}{2\pi i }\frac{e^{\frac{t x}{4} - \frac{x}{2}\left( s - \frac{1}{s} \right)}s^\nu \cdot \chi_{\gamma}(t)\chi_{\hat \gamma}(s) - e^{\frac{x}{2} \left(\frac{s}{2} -\frac{1}{s} \right)}s^{-\nu} \cdot \chi_{\gamma}(s)\chi_{\hat \gamma}(t)}{t-s} \nonumber \\
= \frac{\vec{f}(t)^T \cdot \vec{g}(s)}{t-s}
\end{gather}
with
\begin{align}
\vec{f}(t) &= \frac{1}{2\pi i} \left[ \begin{array}{c} \displaystyle e^{\frac{tx}{4}} \\ 0 \end{array}
\right] \chi_{\gamma}(t)  +  \frac{1}{2\pi i}\left[ \begin{array}{c} 0 \\ \displaystyle 1\end{array} \right] \chi_{\hat \gamma}(t) \\
\vec{g}(s) &= \left[ \begin{array}{c} \displaystyle e^{ \frac{x}{2} \left( -s + \frac{1}{s} \right)} s^\nu \\  0 \end{array} \right] \chi_{\hat \gamma}(s) + \left[ \begin{array}{c} 0 \\ \displaystyle- e^{\frac{x}{2} \left(\frac{s}{2}-\frac{1}{s} \right)}s^{-\nu} \end{array}  \right]  \chi_{\gamma}(s)
\end{align}
\end{subequations}
\end{thm}
\begin{figure}
\centering
\resizebox{.8\textwidth}{!}{
\input{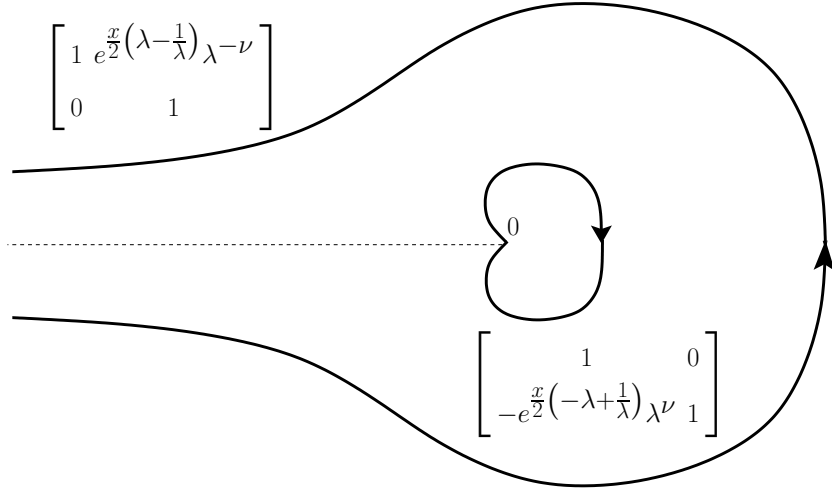}
}
\label{Bessel1timejumps}
\caption{The jump matrices for the Bessel-RHP in the single-time case.}
\end{figure}
The associated $2\times 2$ Riemann-Hilbert problem has jump matrix $M(\lambda):= I - J(\lambda)$ on $\Sigma := \gamma \cup \hat \gamma$ with 
\begin{equation}
J (\lambda) = \left[ \begin{array}{cc} 0 & - e^{\frac{x}{2} \left(\lambda  - \frac{1}{\lambda} \right)} \lambda^{-\nu}\\
0 & 0 \end{array} \right] \chi_{\gamma}(\lambda) + 
\left[ \begin{array}{cc} 0 & 0 \\
e^{\frac{x}{2} \left( -\lambda  + \frac{1}{\lambda} \right)}\lambda^\nu & 0 \end{array} \right] \chi_{\hat \gamma}(\lambda)
\end{equation}

Clearly 
\begin{gather}
M(\lambda) = e^{T(\lambda)} M_0 e^{-T(\lambda)} \nonumber \\
\text{with} \  \  T(\lambda) := \frac{\theta_x}{2} \sigma_3, \ \ \ 
\theta_x := \frac{x}{2}\left( \lambda - \frac{1}{\lambda} \right) - \nu \ln \lambda
\end{gather}
where $M_0$ is a constant matrix. 
Thus, the matrix $\Psi(\lambda):= \Gamma(\lambda) e^{T_x(\lambda)}$ solves a Riemann-Hilbert problem with constant jumps and it is (sectionally) a solution to a polynomial ODE.

Applying again Theorem \ref{JMUtau} and Jimbo-Miwa-Ueno residue formula, we get
\begin{gather} 
\partial_x \ln \det (\text{Id} - \mathbb{B}) = \int_\Sigma \operatorname{Tr} \left( \Gamma_-^{-1}(\lambda) \Gamma'_-(\lambda) \Xi_x(\lambda) \right) \frac{d\lambda}{2\pi i}  \nonumber \\
=  -\underset{\lambda = \infty}{\operatorname{res}} \operatorname{Tr}\left( \Gamma^{-1} \Gamma' \partial_x T \right) +\underset{\lambda = 0}{\operatorname{res}} \operatorname{Tr}\left( \Gamma^{-1} \Gamma' \partial_x T \right). \label{residuex}
\end{gather}

\begin{prop}
The Fredhom determinant of the single-interval Bessel operator satisfies the following identity
\begin{equation}
\partial_x \ln \det \left( \operatorname{Id} - \mathbb{B}\right) = -\frac{1}{2}\Gamma_{1;22} + \frac{1}{2}\tilde\Gamma_{1;2,2} 
\end{equation}
where $\Gamma_{1;22}$ is the $(2,2)$-entry of the residue matrix at infinity, while $ \tilde\Gamma_{1;2,2}$ is the $(2,2)$-entry of residue matrix at zero.
\end{prop}
\begin{proof}
As in the proof of Proposition \ref{gamma122}, we can easily get the result by calculating the derivative of the conjugation matrix
\begin{equation}
\partial_x T(\lambda) =\frac{1}{2} \left( \lambda  - \frac{1}{\lambda}\right) \left( \frac{1}{2}I - E_{2,2} \right)
\end{equation}
and by keeping into account that, since $\det \Gamma (\lambda ) \equiv 1$, $ \operatorname{Tr} \Gamma_1 = \operatorname{Tr} \tilde\Gamma_1 =0$.
\end{proof}

Keeping into account the asymptotic behaviour at infinity of the matrix $\Psi$, we can calculate the Lax pair associated to our Riemann-Hilbert problem.
\begin{align}
A &:= \partial_\lambda\Psi \cdot \Psi^{-1} (\lambda) = A_0 + \frac{A_{-1}}{\lambda} + \frac{A_{-2}}{\lambda^{2}} \label{laxA} \\
U &:= \partial_x\Psi \cdot \Psi^{-1}(\lambda) = U_0 + \lambda U_1 + \frac{U_{-1}}{\lambda}  \label{laxU}
\end{align}
with coefficients
\begin{equation}
\renewcommand\arraystretch{1.9}
\begin{array}{ll} 
A_0 &= \displaystyle \frac{x}{4}\sigma_3 \\
A_{-1} &= \displaystyle \frac{x}{4} \left[ \Gamma_1,\sigma_3 \right] - \frac{\nu}{2}\sigma_3 \\
A_{-2} &= \displaystyle \frac{x}{4}\left[ \Gamma_2,\sigma_3 \right]  + \frac{x}{4}\left[ \sigma_3\Gamma_1,\Gamma_1 \right] - \frac{\nu}{2} \left[ \Gamma_1,\sigma_3 \right] + \frac{x}{4}\sigma_3 - \Gamma_1 \\
U_1 &= \displaystyle \frac{1}{4}\sigma_3 \\
U_0 &= \displaystyle \frac{1}{4} \left[ \Gamma_1,\sigma_3 \right] \\
U_{-1}&= \displaystyle \frac{1}{4} \left( \left[ \Gamma_2,\sigma_3 \right] + \left[ \sigma_3 \Gamma_1,\Gamma_1 \right] - \sigma_3 + 4\dot \Gamma_1 \right)
\end{array}
\end{equation}
These coefficients matches with the results in \cite{Ptrans}.

Adapting the notation used in \cite{Ptrans} to our case, we have
\begin{gather}
A_0= \frac{x}{4}\sigma_3,   \ \ \ \ A_{-1}= \left[ \begin{array}{cc} 
-\frac{\nu}{2}& Y(x) \\ T(x)& \frac{\nu}{2}
\end{array}\right], \ \ \ \ 
A_{-2}= \left[ \begin{array}{cc} 
\frac{x}{4} - L(x)& -W(x) L(x) \\ \frac{L(x) - \frac{x}{2}}{W(x)}& - \frac{x}{4}+ L(x)
\end{array}\right] \nonumber \\
U_1 = \frac{1}{4}\sigma_3 \ \ \ \ U_{0}= \left[ \begin{array}{cc} 
0& \frac{Y(x)}{x} \\ \frac{T(x)}{x} & 0
\end{array}\right], \ \ \ \ 
U_{-1}= \left[ \begin{array}{cc} 
-\frac{1}{4} + \frac{L(x)}{x}& \frac{W(x)L(x)}{x} \\ \frac{\frac{x}{2} - L(x)}{xW(x)} & \frac{1}{4} - \frac{L(x)}{x}
\end{array}\right]
\end{gather}
Finally, calculating the compatibility equation, we get the following system of ODEs
\begin{subequations}
\begin{align}
& \frac{dL}{dx} = \frac{L}{x} + \frac{2WLT}{x} + \frac{2YL}{Wx} - \frac{Y}{W}\\
& \frac{dT}{dx} = \frac{L}{W} - \frac{x}{2W} + \frac{\nu T}{x}\\
& \frac{dW}{dx} = -\frac{2W^2T}{x} + \frac{2Y}{x} + \frac{\nu W}{x}\\
& \frac{dY}{dx} = -\frac{\nu Y}{x} + WL
\end{align}
and the constant of motion 
\begin{equation}
\frac{\Theta_0}{2} = \frac{2\nu L}{x} - \frac{\nu}{2} + \frac{2YL}{Wx} - \frac{Y}{W} - \frac{2WLT}{x}
\end{equation}
\end{subequations}
which can be proven to be the monodromy exponent at $0$ and equal to $-\nu$. 

Setting now
\begin{equation}
F(x) := - \frac{Y(x)}{W(x)L(x)}
\end{equation}
and substituting in the equations above, we get that $F$ and $L$ satisfy
\begin{align}
& x\frac{dF}{dx} = 4LF^2 + xF^2 + (2\nu +1) F + x \\
& x \frac{dL}{dx} = (2\nu +1)L - 4L^2F + 2xFL.
\end{align}

\begin{oss}
The latter equation for the function $L$ is a Bernoulli 1st-order ODE with $n=2$ and it can be explicitly integrated.
\end{oss}

Moreover, differentiating it and using the ODEs above, we get the following Painlev\'e III equation:
\begin{equation}
\frac{d^2F}{dx} = \frac{1}{F}\left(\frac{dF}{dx}\right)^2 - \frac{1}{x} \frac{dF}{dx} - \frac{4}{x} \left( \Theta_0 F^2 + \frac{1}{2} + \frac{1}{2}\nu \right) + F^3 - \frac{1}{F}.
\end{equation}

Given the expression of the matrix $A$, we can find an expression for the residue matrix $\Gamma_1$ and in particular we have
\begin{gather}
\underset{\lambda = \infty}{\operatorname{res}}\operatorname{Tr} \left( \Gamma^{-1} \Gamma' \partial_xT \right)
=  \frac{1}{x}F^2L^2 -\frac{1}{2} F^2L - \frac{\nu}{x}FL +\frac{x}{8}-\frac{L}{2}
\end{gather}

Focusing now on the reside at $0$, we can perform similar calculation with the already known Lax pair (\ref{laxA})-(\ref{laxU}). 
\begin{gather}
\underset{\lambda = 0}{\operatorname{res}}  \operatorname{Tr} \left( \Gamma^{-1} \Gamma' \partial_xT \right) 
= - \frac{1}{x}F^2L^2 +\frac{1}{2} F^2L + \frac{\nu}{x}FL -\frac{x}{8}+\frac{L}{2}
\end{gather}

In conclusion, 
\begin{thm}
The following identity holds
\begin{equation}
\det (\operatorname{Id} - \mathbb{B}) = \operatorname{exp} \, \left\{ \int_0^x H_{\operatorname{III}}(u) \, du \right\}
\end{equation}
where $H_{\operatorname{III}}$ is the Hamiltonian associated to the Painlev\'e III equation (see  \cite{JMUII})
\begin{equation}
H_{\operatorname{III}} (F, L; x) = \frac{1}{x} \left[ -2F^2L^2 + \left( xF^2 + 2\nu F + x \right)L - \frac{x^2}{4} \right]
\end{equation}
\end{thm} 

\begin{oss} The Hamiltonian $H_{\operatorname{III}}$ might be singular at $0$, but by definition
$ \partial \ln \tau = H_{\operatorname{III}}$ and $\tau$ is continuous with $\tau(0) = \det (\operatorname{Id}) = 1$,
thus it is integrable in a (right) neighbourhood of the origin.
\end{oss}

\begin{figure}
\centering
\includegraphics[width=.8\textwidth]{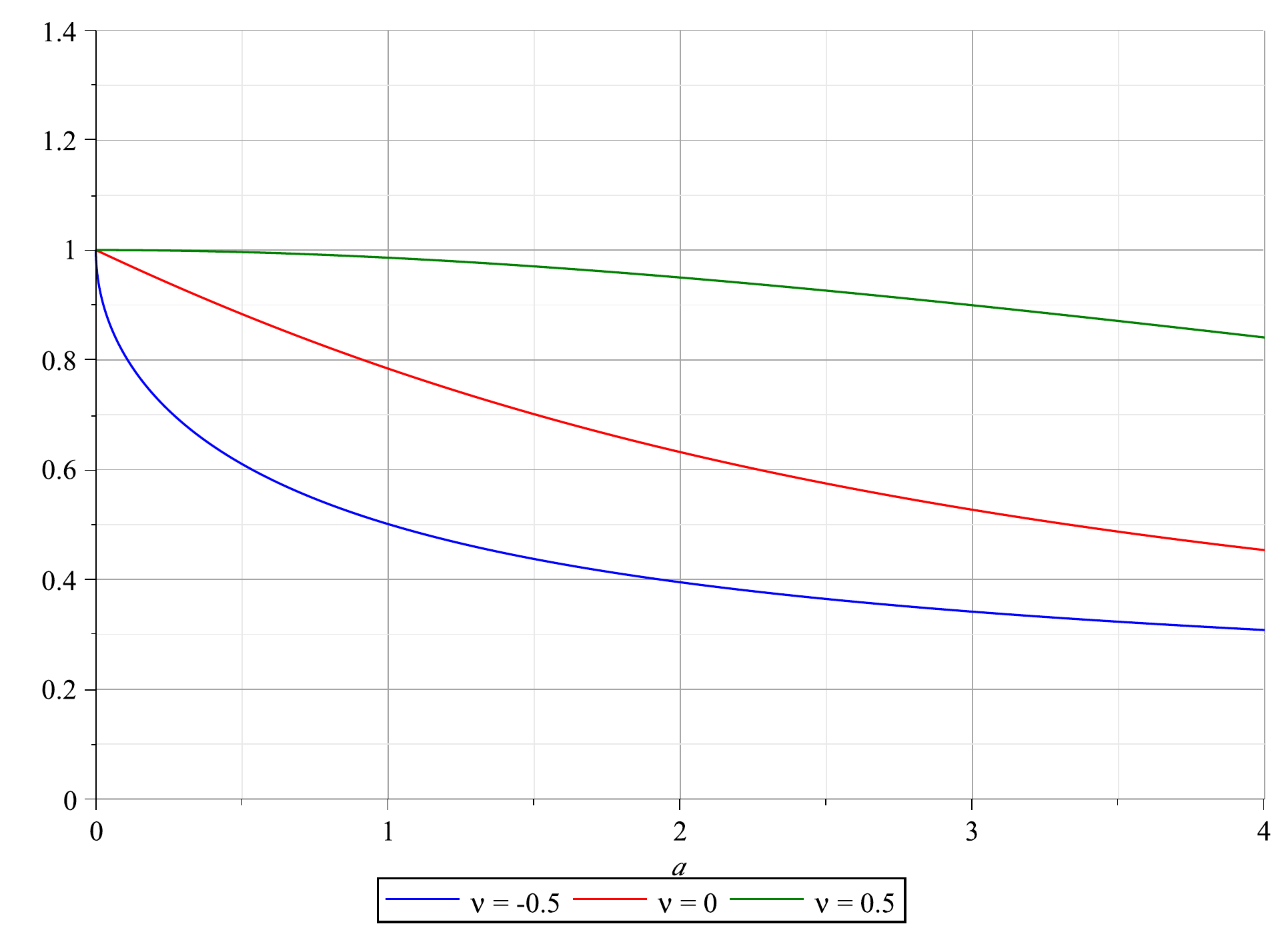}
\label{FredBes}
\caption{Numerical simulation of the Fredholm determinant $\det (\text{Id} - K_B\chi_{[0,a]})$ as a function of $a$, in the case $\nu = -0.5$, $\nu=0$ and $\nu= 0.5$.}
\end{figure}

\section{Gap probabilities for the multi-time Bessel process}\label{gapprobmultiB}

The multi-time Bessel process on $L^2(\mathbb{R_+})$ with times $\tau_1< \ldots < \tau_n$ is governed by the matrix operator $K_B := \tilde K_B + H_B$ with kernels $K_B$, $\tilde K_B(x,y)$ and $H_B(x,y)$ given as follows
\begin{subequations}
\begin{align}
K_{B;ij}(x,y):= & \tilde K_{B;ij}(x,y) + H_{B;ij}(x,y)\\
\tilde K_{B;ij}(x,y):=& 
\frac{1}{(2\pi i)^2} \left(\frac{y}{x}\right)^{\frac{\nu}{2}} \iint_{\gamma \times\hat \gamma_j} \frac{dt \, ds}{ts} \, \frac{e^{ \Delta_{ij}+ xt - \frac{1}{4t} -ys + \frac{1}{4s} }}{ \frac{1}{4t}-  \frac{1}{4s} - \Delta_{ij}} \left( \frac{s}{t}\right)^\nu  \\
H_{B;ij}(x,y) :=& \frac{1}{\Delta_{ji}} \left( \frac{y}{x}\right)^{\frac{\nu}{2}}  \int_\gamma e^{\frac{x}{4\Delta_{ji}}(t-1) + \frac{y}{4\Delta_{ji} }\left( \frac{1}{t} -1 \right)} t^{-\nu -1} \frac{dt}{2\pi i}
\end{align}
\end{subequations}
with the same curve $\gamma$ as in the single-time Bessel kernel (a contour that winds around zero counterclockwise an extends to $-\infty$) and  $\hat \gamma_j := \frac{1}{\gamma + 4\tau_j} $, $\forall \, j=1,\ldots, n$. 
\begin{oss}
The matrix $H_{B;ij}$ is strictly upper triangular. 
\end{oss}

As in the single-time case, we are interested in the following quantity
\begin{equation}
\det\left( \text{Id} - K_B \chi_{\mathcal{I}} \right)
\end{equation}
which is equal to the gap probability of the multi-time Bessel kernel restricted to a collection of multi-intervals $\mathcal{I} = \{ I_1, \ldots, I_n \}$,
\begin{equation*}
I_j:= [a_1^{(j)}, a_2^{(j)}] \cup \ldots \cup [a_{k_j-1}^{(j)}, a_{k_j}^{(j)}].
\end{equation*}

\begin{oss}
The multi-time Bessel operator fails to be trace-class on infinite intervals.
\end{oss}

For the sake of clarity, we will focus on the simple case $I_j= [0, a^{(j)}]$, $j=1,\ldots, n$.  The general case follows the same guidelines described below, but it involves calculations with complicated notation which we prefer to avoid.

\begin{thm}
The following identity between Fredholm determinants holds
\begin{equation}
\det \left( \operatorname{Id} -K_{B}\chi_{\mathcal{I}} \right) = \det \left( \operatorname{Id} - \mathbb{K}_B  \right)
\end{equation}
with $\chi_{\mathcal{I}} = \operatorname{diag} \, \left(\chi_{I_1}, \ldots, \chi_{I_n} \right)$
the characteristic matrix of the collection of multi-intervals. The operator $\mathbb{K}_B$ is an integrable operator with a $2n\times 2n$ matrix kernel  of the form
\begin{equation}
\mathbb{K}_B (t,\xi) = \frac{\textbf{f}(t)^T \cdot \textbf{g}(\xi)}{t-\xi}  \label{IIKSBesselmulti1}
\end{equation}
acting on the Hilbert space
\begin{equation}
H: =L^2\left(\gamma \cup \bigcup_{k=1}^n \gamma_{-k} ,\mathbb{C}^n\right) \sim  L^2\left(\bigcup_{k=1}^n\gamma_{-k}, \mathbb{C}^n\right) \oplus  L^2(\gamma, \mathbb{C}^n),
\end{equation}
with $\gamma_{-k} := \frac{1}{\gamma}-4\tau_k$.

The functions $\textbf{f}, \, \textbf{g} $ are the following $2n\times 2n$ matrices
\begin{gather}
\textbf{f}(t) =  \left[ \begin{array}{c|c}
\operatorname{diag} \, \mathcal{N} (t) & 0  \\ \hline
0 & A(\mathcal{M}(t))   \\ \hline
0 & B(\mathcal{H}(t))
\end{array} \right] \\
\textbf{g}(\xi) = \left[ \begin{array}{c|c}
0 &  \operatorname{diag} \, \mathcal{N}(\xi) \\ \hline
C(\mathcal{M}(\xi)) & 0 \\ \hline
0 &  D(\mathcal{H}(\xi))
 \end{array} \right] 
 \end{gather}
where $\operatorname{diag} \, \mathcal{N}$ is a $n\times n$ matrix, $A$ and $C$ are two rows with $n$ entries and $B$ and $D$ are $(n-1)\times n$ matrices, 
\begin{gather*}
\operatorname{diag} \, \mathcal{N}(t) := \operatorname{diag} \left( - 4  e^{  - \frac{a^{(1)}}{t_1} }  \chi_{\gamma} (t), \ldots,  - 4  e^{ - \frac{a^{(n)}}{t_n} }  \chi_{\gamma}(t)  \right) \\
A(\mathcal{M}(t)) := \left[ e^{ - \frac{t}{4} } t_1^{\nu} \chi_{\gamma_{-1}}(t), \ldots,  e^{ - \frac{t}{4} } t_n^{\nu} \chi_{\gamma_{-n}}(t) \right] \\
B(\mathcal{H}(t)) := \left[ \begin{array}{ccccc}
-4 e^{-\frac{a^{(2)}}{t_2}} \frac{t_1^\nu}{t_2^\nu}\chi_{\gamma_{-1}}  & 0 & &  & \\
-4 e^{-\frac{a^{(3)}}{t_3}} \frac{t_1^\nu}{t_3^\nu}\chi_{\gamma_{-1}}  & -4 e^{-\frac{a^{(3)}}{t_3}} \frac{t_2^\nu}{t_3^\nu}\chi_{\gamma_{-2}} && & \\
\vdots & \vdots & \ddots & & \\
-4 e^{-\frac{a^{(n)}}{t_n}} \frac{t_1^\nu}{t_n^\nu}\chi_{\gamma_{-1}} &  -4 e^{-\frac{a^{(n)}}{t_n}} \frac{t_2^\nu}{t_n^\nu}\chi_{\gamma_{-2}}  &&  -4 e^{-\frac{a^{(n)}}{t_n}} \frac{t_{n-1}^\nu}{t_n^\nu}\chi_{\gamma_{-{(n-1)}}} & 0
\end{array} \right] \\
\operatorname{diag}\, \mathcal{N}(\xi):=  \operatorname{diag} \left( e^{  \frac{a^{(1)}}{\xi_1}} \chi_{\gamma_{-1}} (\xi), \ldots, e^{  \frac{a^{(n)}}{\xi_n}} \chi_{\gamma_{-n}}(\xi)  \right)  \\
C(\mathcal{M}(\xi):= \left[ e^{  \frac{\xi}{4}} \xi_1^{-\nu} \chi_{\gamma}(\xi), \ldots, e^{  \frac{\xi}{4}} \xi_n^{-\nu} \chi_{\gamma}(\xi)  \right]  \\
D(\mathcal{H}(\xi)) =
\left[ \begin{array}{cccccc}
0 & e^{\frac{a^{(2)}}{\xi_2}}\chi_{\gamma_{-2}}(\xi) &&&& \\
& 0 & e^{\frac{a^{(3)}}{\xi_3}}\chi_{\gamma_{-3}}(\xi) &&& \\
&& 0 & e^{\frac{a^{(4)}}{\xi_4}}\chi_{\gamma_{-4}}(\xi) && \\
&&&& \ddots & \\
&&&& 0 & e^{\frac{a^{(n)}}{\xi_n}}\chi_{\gamma_{-n}} (\xi)
\end{array} \right] \label{IIKSBesselmulti2}
\end{gather*}
with $\xi_k := \xi + 4\tau_k $, $t_k := t + 4\tau_k$, for $k=1,\ldots, n$.
\end{thm}

\begin{oss}
The naming of Fredholm determinant in the theorem above needs some clarification: by $``\det"$ we denote the determinant defined through the Fredholm expansion
\begin{equation}
\det(\operatorname{Id}-K) := 1+ \sum_{k=1}^\infty \frac{1}{k!}\int_{X^k} \det [K(x_i,x_j)]_{i,j=1}^k d\mu(x_1)\ldots d\mu(x_k) \label{Freddet}
\end{equation} 
with $K$ an integral operator acting on the Hilbert space $L^2(X, d\mu(x))$, with kernel $K(x,y)$.

In our case, the operator $K_B\chi_{\mathcal{I}}= (\tilde K_B + H_B)\chi_{\mathcal{I}}$ is actually the sum of a trace-class operator ($\tilde K_B \chi_{\mathcal{I}}$) and a Hilbert-Schmidt operator ($H_B\chi_{\mathcal{I}}$) whose kernel is diagonal-free, as it will be clear along the proof.

Thus, to be precise, we have the following chain of identities
\begin{gather}
``\det" (\operatorname{Id} - K_B\chi_{\mathcal{I}}) = ``\det" (\operatorname{Id} - \tilde K_B\chi_{\mathcal{I}} - H_B\chi_{\mathcal{I}}) \nonumber \\
= e^{\operatorname{Tr} \tilde K_B} \left.\det\right._2(\operatorname{Id} - \tilde K_B\chi_{\mathcal{I}} - H_B\chi_{\mathcal{I}})
\end{gather}
where $\left. \det \right._2$ denotes the regularized Carleman determinant (see \cite{traceideals} for a detailed description of the theory).
\end{oss}

\begin{proof}
Thanks to the invariance of the Fredholm determinant under kernel conjugation, we can discard the term $\left( \frac{y}{x} \right)^{\nu/2}$ in our further calculations.

We will work on the entry $(i,j)$ of the kernel. We can notice that for $x<0$ or $y<0$ the kernel is identically zero, $K_B(x,y) \equiv 0$. Then, applying Cauchy's theorem and after some suitable calculations, we have
\begin{gather}
\chi_{[0,a^{(j)}]}(y) \cdot \tilde K_{B;ij}(x,y)   \nonumber \\
= \int_{i\mathbb{R} + \epsilon} \frac{d\xi}{2\pi i} \frac{e^{\xi(a^{(j)}-y)}}{\xi - s} \iint_{\gamma \times\hat \gamma_j} \frac{dt \, ds}{(2\pi i)^2 ts} \, \frac{e^{ \Delta + xt - \frac{1}{4t} -a^{(j)}s + \frac{1}{4s} }}{ \frac{1}{4t}-  \frac{1}{4s} - \Delta} \left( \frac{s}{t}\right)^\nu \nonumber \\
=  \int_{i\mathbb{R} + \epsilon} \frac{d\xi}{2\pi i} e^{-\xi y}  \int_{ i\mathbb{R}+\epsilon} \frac{dt}{2\pi i}e^{xt}  \int_{\hat \gamma_j} \frac{ds}{2\pi i} \, \frac{e^{ \Delta + \xi a^{(j)}- \frac{1}{4t} -a^{(j)}s + \frac{1}{4s} }}{ (\xi - s) \left(\frac{1}{4t}-  \frac{1}{4s} - \Delta\right)} \left( \frac{s}{t}\right)^\nu \frac{1}{ts} \nonumber \\
=  \int_{i\mathbb{R} + \epsilon} \frac{d\xi}{2\pi i} e^{-\xi y}  \int_{ i\mathbb{R}+\epsilon} \frac{dt}{2\pi i}e^{xt}
  \int_{\gamma } \frac{ds}{2\pi i} \, \frac{4e^{ \tau_i + \xi a^{(j)}- \frac{1}{4t} -\frac{a^{(j)}}{s+ 4\tau_j} + \frac{s}{4} }}{ \left(\frac{1}{\xi} -4\tau_j -s\right) \left(\frac{1}{t}- 4\tau_i -s \right)} \left( \frac{1}{(s+4\tau_j)t}\right)^{\nu} \frac{1}{t\xi}
\end{gather}
where we deformed $\gamma$ into a translated imaginary axis $i\mathbb{R}+\epsilon$ ($\epsilon>0$) in order to make Fourier operator defined below more explicit; the last equality follows from the change of variable on $s$: $s\rightarrow 1/(s+4\tau_j)$ (thus the contour $\hat \gamma_j$ becomes similar to and can be continuously deformed into $\gamma$). 
 
On the other hand
\begin{gather}
\chi_{[0, a^{(j)}]}(y) \cdot H_{B;ij}(x,y) \nonumber \\
=\frac{-1}{\Delta_{ji}}  \int_{i\mathbb{R}+\epsilon} \frac{d\xi}{2\pi i} \frac{e^{\xi(a^{(j)} -y)}}{\xi - \frac{1}{4\Delta_{ji}}\left(1-\frac{1}{t}\right)}  \int_\gamma e^{\frac{x}{4\Delta_{ji}}(t-1) - \frac{a^{(j)}}{4\Delta_{ji} }\left(1- \frac{1}{t} \right)} t^{-\nu -1} \frac{dt}{2\pi i} \nonumber \\
= \frac{-1}{\Delta_{ji}}  \int_{i\mathbb{R}+\epsilon} \frac{d\xi}{2\pi i} e^{-\xi y}  \int_{i\mathbb{R}+\epsilon} \frac{e^{\xi a^{(j)} + \frac{x}{4\Delta_{ji}}(t-1) - \frac{a^{(j)}}{4\Delta_{ji} }\left(1- \frac{1}{t} \right)}}{\xi - \frac{1}{4\Delta_{ji}}\left(1-\frac{1}{t}\right)} t^{-\nu -1} \frac{dt}{2\pi i} \nonumber \\
=  -4 \int_{i\mathbb{R}+\epsilon} \frac{d\xi}{2\pi i} e^{-\xi y}  \int_{i\mathbb{R}+ \epsilon}  e^{xt} \frac{dt}{2\pi i}  \,   \frac{e^{ a^{(j)} \left(\xi - \frac{t}{4\Delta_{ji}t+1 }\right)}}{t \xi \left( 4\Delta_{ji} + \frac{1}{t} -\frac{1}{\xi} \right)} (4\Delta_{ji}t+1)^{-\nu}
\end{gather}

It is easily recognizable the conjugation with a Fourier-like operator as in (\ref{Fourier}), so that
\begin{equation}
\left(K_B\chi_{\mathcal{I}}\right)_{ij} = \mathcal{F}^{-1} \circ \left( \mathcal{B}_{ij} + \chi_{i<j}\mathcal{H}_{ij} \right) \circ \mathcal{F}
\end{equation}
with
\begin{gather}
\mathcal{B}_{ij}(t,\xi)=  \int_{\gamma } \frac{ds}{2\pi i} \, \frac{4e^{ \tau_i + \xi a^{(j)}- \frac{1}{4t} -\frac{a^{(j)}}{s+ 4\tau_j} + \frac{s}{4} }}{ \left(\frac{1}{\xi} -4\tau_j -s\right) \left(\frac{1}{t}- 4\tau_i -s \right)} \left( \frac{1}{(s+4\tau_j)t}\right)^{\nu} \frac{1}{t\xi} \\
\mathcal{H}_{ij}(t,\xi):= -4  \frac{e^{ a^{(j)} \left(\xi - \frac{t}{4\Delta_{ji}t+1 }\right)}}{  4\tau_{j} - 4\tau_i + \frac{1}{t} -\frac{1}{\xi}} (4\Delta_{ji}t+1)^{-\nu} \frac{1}{\xi t}
\end{gather}

Now we will perform a change of variables on the Fourier-transformed kernel $\mathcal{B}_{ij}+ \chi_{i<j}\mathcal{H}_{ij}$: $\xi_j:= \frac{1}{\xi} - 4\tau_j$ and $\eta_i:= \frac{1}{t} -4 \tau_i$. 
This will lead to the following expression for the (Fourier-transformed) multi-time Bessel kernel
\begin{gather}
 \mathcal{K}_{B; ij} =  \mathcal{B}_{ij}(\eta, \xi) + \chi_{\tau_i<\tau_j}\mathcal{H}_{ij} (\eta, \xi) = \nonumber \\
4 \int_{\gamma} \frac{dt}{2\pi i} \, \frac{e^{ \frac{ a^{(j)}}{\xi+4\tau_j}- \frac{\eta}{4} -\frac{a^{(j)}}{t+4\tau_j} + \frac{t}{4} }}{ \left(\xi -t\right) \left(\eta -t \right)} \left( \frac{\eta + 4\tau_i}{t+4\tau_j}\right)^{\nu}
+ \chi_{\tau_i<\tau_j} \cdot 4  \frac{e^{ \frac{a^{(j)}}{\xi + 4\tau_j} - \frac{a^{(j)}}{\eta + 4\tau_j }}}{\xi-\eta} \left(\frac{\eta + 4\tau_j }{\eta + 4\tau_i}\right)^{-\nu}  \label{FTBesselmulti}
\end{gather}
with $\xi \in \frac{1}{\gamma} - 4\tau_j=: \gamma_{-j}$ and $\eta \in \frac{1}{\gamma} - 4\tau_i=: \gamma_{-i} $, $\forall \, i,j =1,\ldots,n$. Such operator is acting on the Hilbert space $L^2\left(\bigcup_{k=1}^n \gamma_{-k}, \mathbb{C}^n \right) \sim \bigoplus_{k=1}^n L^2\left( \gamma_{-k}, \mathbb{C}^n \right)$.

\begin{lemma} The operator $\mathcal{B}$ is trace-class and the operator $\mathcal{H}$ is Hilbert-Schmidt. Moreover, the following decomposition holds $\mathcal{K}_B = \mathcal{M} \circ \mathcal{N} + \mathcal{H}$, where
\begin{equation}
\renewcommand\arraystretch{2}
\begin{array}{l}
\mathcal{M}: \  L^2(\gamma, \mathbb{C}^n) \rightarrow L^2\left(\bigcup_{k=1}^n \gamma_{-k}, \mathbb{C}^n \right) \\
\mathcal{N}: \  L^2\left(\bigcup_{k=1}^n \gamma_{-k}, \mathbb{C}^n \right)  \rightarrow L^2(\gamma, \mathbb{C}^n) \\
\mathcal{H}: \ L^2\left(\bigcup_{k=1}^n \gamma_{-k}, \mathbb{C}^n \right)  \rightarrow L^2\left(\bigcup_{k=1}^n  \gamma_{-k}, \mathbb{C}^n \right) 
\end{array}
\end{equation}
with entries
\begin{subequations}
\begin{gather}
\mathcal{M}_{ij}(t, \eta) :=  \frac{1}{2\pi i}\frac{ e^{ - \frac{\eta}{4} + \frac{t}{4}}}{ \eta - t  }  \left( \frac{\eta_i}{t_j} \right)^{\nu} \cdot \chi_{\gamma}(t) \cdot \chi_{\gamma_{-i}}(\eta)\\
\mathcal{N}_{ij}(\xi, t ; a^{(j)}) = 4 \delta_{ij} \cdot  \frac{ e^{  a^{(j)}\left(  \frac{1}{\xi_j} - \frac{1}{t_j}\right) }}{ \xi - t } \cdot \chi_{\gamma_{-j}}(\xi) \cdot \chi_{\gamma}(t) \\
\mathcal{H}_{ij} (\xi,\eta) =  \chi_{\tau_i<\tau_j} \cdot 4  \frac{e^{ a^{(j)} \left(\frac{1}{\xi_j} - \frac{1}{\eta_j}\right)}}{\xi-\eta} \left(\frac{\eta_j}{\eta_i}\right)^{-\nu}  \cdot \chi_{\gamma_{-i}}(\eta) \cdot \chi_{\gamma_{-j}}(\xi)
\end{gather}
\end{subequations}
$\zeta_k := \zeta+4\tau_k$ ($\zeta = \xi, t, \eta$) and $\gamma_{-k}:= \frac{1}{\gamma}-4\tau_k$, $\forall \, k=1, \ldots, n$.

\end{lemma}

\begin{proof}
All the kernels are of the general form $H(z,w)$ with $z$ and $w$ on disjoint supports, that we indicate now temporaritly  by $S_1,S_2$. It is then simple to see that in each instance $\int_{S_1} \int_{S_2} |H(z,w)|^2 |d z| |d w|<+\infty$ and hence each operator is Hilbert-Schmidt.  Then $\mathcal B$ is trace class because it is the composition of two HS operators.
\end{proof}
Now consider the Hilbert space 
\begin{equation}
H: =L^2\left(\gamma \cup \bigcup_{k=1}^n \frac{1}{\gamma}-4\tau_k ,\mathbb{C}^n\right) \sim  L^2\left(\bigcup_{k=1}^n \frac{1}{\gamma}-4\tau_k, \mathbb{C}^n\right) \otimes L^2(\gamma, \mathbb{C}^n),
\end{equation}
and the matrix operator $\mathbb{K}_B: H \rightarrow H$ defined as 
\begin{equation}
\mathbb{K}_B = \left[ \begin{array}{c|c} 
0 & \mathcal{N} \\ \hline
\mathcal{M} & \mathcal{H}
\end{array} \right]
\end{equation}
due to the splitting of the space $H$ into its two main addenda.

For now, we denote by $``\det"$ the determinant defined by the Fredholm expansion (\ref{Freddet}); then, $``\det"(\text{Id} - \mathbb{K}_B) = \left. \det\right._2(\text{Id} - \mathbb{K}_B)$, since its kernel is diagonal-free. Moreover, we introduce another Hilbert-Schmidt operator
\begin{equation*}
\mathbb{K}'_B = \left[ \begin{array}{c|c} 
0 & -\mathcal{N} \\ \hline
0&0
\end{array} \right]
\end{equation*}
which is only Hilbert-Schmidt, but nevertheless its Carleman determinant ($ \text{det}_2$) is well defined and $\left. \det \right._2(I - \mathbb{K}'_B) \equiv 1$.

Collecting all the results we have seen so far, we perform the following chain of equalities
\begin{gather}
``\det" (\text{Id}_{L^2(\mathbb{R}_+)} - K_B \chi_{\mathcal{I}}) 
= \left. \det\right._2 \left( \text{Id} - K_B \chi_{\mathcal{I}}\right)e^{-\operatorname{Tr}(\tilde K)}  \nonumber \\
= \left. \det\right._2 \left( \text{Id}_{L^2\left(\bigcup_{k=1}^n \gamma_{-k}\right)} - \mathcal{K}_B \right)e^{- \operatorname{Tr} (\mathcal{B})} 
= \left. \det \right._2 (\text{Id}_H - \mathbb{K}_B) \left. \det \right._2(\text{Id}_H - \mathbb{K}'_B) \nonumber \\
= \left. \det \right._2(\text{Id}_H - \mathbb{K}_B) 
= `` \det" (\text{Id}_H - \mathbb{K}_B) 
\end{gather}
The first equality follows from the fact that $K_B - \tilde K_B$ is diagonal-free; the second equality follows from invariance of the determinant under Fourier transform; the first identity on the last line is just an application of the following result: given $\mathbb{K}_B$, $\mathbb{K}_B'$ Hilbert-Schmidt operators, then
\begin{equation*}
 \left. \det \right._2 (\text{Id} - \mathbb{K}_B) \left. \det \right._2(\text{Id} - \mathbb{K}'_B)  = \left. \det \right._2 (\text{Id} - \mathbb{K}_B - \mathbb{K}_B' + \mathbb{K}_B\mathbb{K}_B') e^{\operatorname{Tr}(\mathbb{K}_B\mathbb{K}_B')} .
 \end{equation*}
 
It is finally just a matter of computation to show that $\mathbb{K}_B$ is an integrable operator of the form (\ref{IIKSBesselmulti1})-(\ref{IIKSBesselmulti2}).
\end{proof}

\bigskip

For the sake of clarity, let us consider a simple example of the multi-time Bessel process with two times $\tau_1, \tau_2$, restricted to the finite intervals $I_1:= [0,a]$ and $I_2:= [0,b]$:
\begin{gather}
K_B(x,y) \cdot \text{diag} \left[ \chi_{[0,a]}(y), \chi_{[0,b]}(y) \right] = \left(\frac{y}{x}\right)^{\frac{\nu}{2}} \times \nonumber \\
\left\{  \left[ \begin{array}{cc}
4 \chi_{[0,a]}(y) \int_{\gamma \times\hat \gamma_j} \frac{dt \, ds}{(2\pi i)^2} \, \frac{e^{ xt - \frac{1}{4t} -ys + \frac{1}{4s} }}{ t-s} \left( \frac{s}{t}\right)^\nu 
& \frac{\chi_{[0,b]}(y)}{(2\pi i)^2} \int_{\gamma \times\hat \gamma_j} \frac{dt \, ds}{ts} \, \frac{e^{ \Delta_{12}+ xt - \frac{1}{4t} -ys + \frac{1}{4s} }}{ \frac{1}{4t}-  \frac{1}{4s} - \Delta_{12}} \left( \frac{s}{t}\right)^\nu\\
\frac{\chi_{[0,a]}(y)}{(2\pi i)^2} \int_{\gamma \times\hat \gamma_j} \frac{dt \, ds}{ts} \, \frac{e^{ \Delta_{21}+ xt - \frac{1}{4t} -ys + \frac{1}{4s} }}{ \frac{1}{4t}-  \frac{1}{4s} - \Delta_{21}} \left( \frac{s}{t}\right)^\nu 
& 4\chi_{[0,b]}(y) \int_{\gamma \times\hat \gamma_j} \frac{dt \, ds}{(2\pi i)^2} \, \frac{e^{ xt - \frac{1}{4t} -ys + \frac{1}{4s} }}{ t-s} \left( \frac{s}{t}\right)^\nu
\end{array} \right] \right. \nonumber \\
\left. + \left[\begin{array}{cc} 0 & -\chi_{[0,b]}(y) \frac{1}{\Delta_{12}}  \int_\gamma e^{\frac{x}{4\Delta_{12}}(1-t) + \frac{y}{4\Delta_{12} }\left( 1- \frac{1}{t}  \right)} t^{-\nu -1} \frac{dt}{2\pi i} \\ 0 &0 \end{array} \right] \right\}
\end{gather}

Then, the integral operator $\mathbb{K}_B:H\rightarrow H$ on the space $H:= L^2\left(\gamma \cup \gamma_{-1} \cup \gamma_{-2}, \mathbb{C}^2\right)$  has the following expression
\begin{gather}
\mathbb{K}_B = \left[ \begin{array}{c|c} 
0 & \mathcal{N} \\ \hline
\mathcal{M} & \mathcal{H}
\end{array} \right] =\nonumber \\
\left[ \begin{array}{cc|cc}
0 & 0 & -4  e^{  \frac{a}{\xi_1} }  \chi_{\gamma_{-1}} e^{  - \frac{a}{t_1} } \chi_{\gamma} & 0 \\
0 & 0 & 0 & -4  e^{  \frac{b}{\xi_2} }  \chi_{\gamma_{-2}} e^{  - \frac{b}{t_2} } \chi_{\gamma} \\ \hline
e^{  \frac{\xi}{4} } \xi_1^{-\nu} \chi_{\gamma}  e^{ - \frac{t}{4}} t_1^{\nu} \chi_{\gamma_{-1}} &  e^{  \frac{\xi}{4} } \xi_2^{-\nu}\chi_{\gamma}  e^{ - \frac{t}{4}} t_1^{\nu} \chi_{\gamma_{-1}}  & 0 & -4  e^{ \frac{b}{\xi_2}} \chi_{\gamma_{-2}}  e^{ - \frac{b}{t_2 }} \frac{t_1 ^\nu}{t_2^\nu} \chi_{\gamma_{-1}}  \\
 e^{  \frac{\xi}{4} } \xi_1^{-\nu} \chi_{\gamma}  e^{  -\frac{t}{4}} t_2^{\nu} \chi_{\gamma_{-2}}&  e^{  \frac{\xi}{4} } \xi_2^{-\nu} \chi_{\gamma}  e^{ - \frac{t}{4}} t_2^{\nu} \chi_{\gamma_{-2}}  & 0 & 0 
\end{array}
\right] 
\end{gather}
and the equality between Fredholm determinants holds
\begin{equation}
\det\left(\text{Id}_{L^2(\mathbb{R}_+, \mathbb{C}^2)} - K_B \text{diag}(\chi_{I_1}, \chi_{I_2})\right) = \det \left( \text{Id}_H - \mathbb{K}_B \right).
\end{equation}

\paragraph{The Riemann-Hilbert problem for the multi-time Bessel process.}

As explained in the introduction, we can relate the computation of the Fredholm determinant of the matrix Bessel operator to the theory of isomonodromic equations, through a suitable Riemann-Hilbert problem.

\begin{prop}
Given the integrable kernel (\ref{IIKSBesselmulti1})-(\ref{IIKSBesselmulti2}), the associated Riemann-Hilbert problem is the following: 
\begin{subequations}
\begin{align}
& \Gamma_+(\lambda)  = \Gamma_-(\lambda) \left( I - 2\pi i J_B(\lambda) \right) \ \ \  \lambda \in \Sigma \\
& \Gamma(\lambda) = I + \mathcal{O}\left( \frac{1}{\lambda} \right) \ \ \ \lambda \rightarrow \infty
\end{align}
\end{subequations}
where $\Gamma$ is a $2n\times 2n$ matrix $\Gamma$ such that it is analytic on the complex plane except at $\Sigma :=  \gamma \cup \bigcup_{k=1}^n \frac{1}{\gamma}-4\tau_k$; the jump matrix $J_B(\lambda) := \textbf{f}(\lambda)\cdot \textbf{g}(\lambda)^T$ has the expression
\begin{gather}
J_B(\lambda) := \left[ \begin{array}{c|c|c}
0 & \star_1 & 0 \\ \hline
\star_2 & 0 & \star_3 \\ \hline
\star_4 & 0 & \star_5
\end{array} \right] \nonumber \\
\star_1 := \left[ -4e^{\theta_1}\chi_\gamma, \ldots, -4e^{\theta_n}\chi_\gamma \right]^T \nonumber \\
\star_2 := \left[ e^{-\theta_1} \chi_{\gamma_{-1}}, \ldots, e^{-\theta_n} \chi_{\gamma_{-n}}  \right] \nonumber \\
\star_3:= \left[e^{-\theta_2} \chi_{\gamma_{-2}}, \ldots, e^{-\theta_n} \chi_{\gamma_{-n}} \right]  \nonumber \\
\star_4:= \left[ \begin{array}{ccccc}
-4e^{- \theta_1+\theta_2} \chi_{\gamma_{-1}}& 0 &&& \\
-4e^{- \theta_1 + \theta_3} \chi_{\gamma_{-1}}& -4e^{-\theta_2 + \theta_3} \chi_{\gamma_{-2}} & 0 && \\
-4e^{- \theta_1 + \theta_4} \chi_{\gamma_{-1}}& -4e^{ -\theta_2+ \theta_4} \chi_{\gamma_{-2}} & -4e^{-\theta_3+ \theta_4 } \chi_{\gamma_{-3}} & 0 & \\
\vdots&&& \ddots& \\
-4e^{-\theta_1+ \theta_n} \chi_{\gamma_{-1}}& \ldots && -4e^{- \theta_{n-1}+ \theta_n} \chi_{\gamma_{-(n-1)}} & 0
\end{array} \right]  \nonumber \\
\star_5:= \left[ \begin{array}{ccccc}
0&&&&  \\
-4e^{-\theta_2+\theta_3}\chi_{\gamma_{-2}} &&& &  \\
-4 e^{-\theta_2 + \theta_4} \chi_{\gamma_{-2}}& -4 e^{-\theta_3+\theta_4} \chi_{\gamma_{-3}} && & \\
 \vdots &\vdots &&& \\ 
 &&&& \\
 -4e^{-\theta_2+\theta_n}\chi_{\gamma_{-2}} && -4e^{-\theta_{n-1}+\theta_n}\chi_{\gamma_{n-1}} & &0
\end{array} \right] \nonumber \\
\theta_i := \frac{\lambda}{4} - \frac{a_i}{\lambda_i} - \nu \ln \lambda_i, \ \ \ \ \  \lambda_i = \lambda + 4\tau_i
\end{gather} 
\end{prop}

We recall that we are considering the simple case $\mathcal{I} = \bigsqcup_j I_j$ with  $I_j:= [0,a^{(j)}]$, $\forall \, j=1,\ldots,n$.

Applying again the results stated in \cite{Misomonodromic} and \cite{MeM}, we can claim the following.
\begin{thm} \label{TEOFREDHOLM}
Given $n$ times $\tau_1 < \tau_2 < \ldots < \tau_n$ and given the multi-interval matrix $\chi_{\mathcal{I}}:= \operatorname{diag} \, \left( \chi_{I_1}, \ldots, \chi_{I_n}\right) $, 
the Fredholm determinant $\det \left( \operatorname{Id} - K_B\chi_{\mathcal{I}} \right)$ is equal to the isomonodromic $\tau$-function related to the above Riemann-Hilbert problem. 

In particular, we have
\begin{subequations}
\begin{align}
&\partial \ln \det \left( \operatorname{Id} -  K_B \chi_{\mathcal{I}} \right) = \int_\Sigma \operatorname{Tr} \left( \Gamma_-^{-1}(\lambda) \Gamma'_-(\lambda)\Xi_\partial (\lambda) \right) \frac{d\lambda}{2\pi i} \\
& \Xi_\partial(\lambda) = 
- 2\pi i \, \partial J_B  \left( I + 2\pi i J_B \right)
\end{align}
\end{subequations}
the $'$ notation means differentiation with respect to $\lambda$, while with $\partial$ we denote any of the derivatives with respect to times $\partial_{\tau_k}$ or endpoints $\partial_{a^{(k)}}$ ($k = 1,\ldots, n $).
\end{thm}


\begin{proof}
Keeping into account formula \ref{Misomonodromictau}, it is enough to verify that $H(M)= 0$ with $M(\lambda)  = I-J_B(\lambda)$.
\end{proof}

\bigskip

In the simple $2$-times case, the jump matrix is
\begin{gather}
J_B(\lambda) = \textbf{f}(\lambda) \cdot \textbf{g}(\lambda)^T = \nonumber \\
\left[ \begin{array}{cc|c|c}
0 & 0 & -4e^{\frac{\lambda}{4}-\frac{a}{\lambda_1}}\lambda_1^{-\nu}\chi_{\gamma} & 0 \\
0 & 0 & -4e^{\frac{\lambda}{4}-\frac{b}{\lambda_2}}\lambda_2^{-\nu} \chi_{\gamma}&  0\\ \hline
e^{-\frac{\lambda}{4} + \frac{a}{\lambda_1}}\lambda_1^\nu \chi_{\gamma_{-1}}& e^{-\frac{\lambda}{4}+\frac{b}{\lambda_2}}\lambda_2^\nu \chi_{\gamma_{-2}} & 0 & e^{-\frac{\lambda}{4}+\frac{b}{\lambda_2}}\lambda_2^\nu \chi_{\gamma_{-2}} \\ \hline
-4 e^{\frac{a}{\lambda_1} - \frac{b}{\lambda_2}} \left( \frac{\lambda_1}{\lambda_2} \right)^\nu\chi_{\gamma_{-1}} & 0 & 0 & 0
\end{array}\right]
\end{gather}

\paragraph{}
Thanks to Theorem \ref{TEOFREDHOLM}, it is possible to derive some more explicit differential identities by using the Jimbo-Miwa-Ueno residue formula (see \cite{Misomonodromic}). 

First we notice that the jump matrix is equivalent up to conjugation with a constant matrix $J_0$:
\begin{gather}
J_B(\lambda)  = e^{T_B(\lambda)} \cdot J_B^0 \cdot e^{-T_B(\lambda)}
\end{gather}
with 
\begin{gather}
T_B(\lambda) := \text{diag}\left[\theta_1 - \frac{\kappa}{2n}, \ldots, \theta_n - \frac{\kappa}{2n}, 1 - \frac{\kappa}{2n}, \theta_2- \frac{\kappa}{2n}, \ldots, \theta_n - \frac{\kappa}{2n} \right]  \nonumber \\
\kappa := \theta_1 + 2\sum_{k=2}^n\theta_k
\end{gather}

Therefore, the matrix $\Psi_B(\lambda) = \Gamma(\lambda)\cdot e^{T_B(\lambda)}$ solves a Riemann-Hilbert problem with constant jumps and it is (sectionally) a solution to a polynomial ODE. 

\begin{thm} Under our previously mentioned hypotheses (see Theorem \ref{TEOFREDHOLM}), we have
\begin{gather}
\int_\Sigma \operatorname{Tr} \left( \Gamma_-^{-1}(\lambda) \Gamma'_-(\lambda)\Xi_\partial (\lambda) \right) \frac{d\lambda}{2\pi i}  
= - \underset{\lambda = \infty}{\operatorname{res}} \operatorname{Tr}\left( \Gamma^{-1}\Gamma' \partial T_B \right) + \nonumber \\
+ \sum_{i=1}^n \underset{\lambda = -4\tau_i}{\operatorname{res}} \operatorname{Tr}\left( \Gamma^{-1}\Gamma' \partial T_B \right) 
\end{gather}
More specifically, regarding the derivative with respect to the endpoints $a^{(i)}$ ($i=1,\ldots,n$), we have
\begin{subequations}
\begin{align}
\underset{\lambda = -4\tau_1}{\operatorname{res}} \operatorname{Tr}\left( \Gamma^{-1}\Gamma' \partial_{a^{(1)}} T_B \right) &=  \left( \frac{1}{2n} -1 \right) \left( \Gamma_{0}^{-1} \Gamma_{1} \right)_{(1,1)} \\
\underset{\lambda = -4\tau_i}{\operatorname{res}}\operatorname{Tr}\left( \Gamma^{-1}\Gamma' \partial_{a^{(i)}} T_B \right)   &= \left(\frac{1}{n}-1\right) \left[ \left( \Gamma_{0}^{-1} \Gamma_{1} \right)_{(i,i)} + \left( \Gamma_{0}^{-1} \Gamma_{1} \right)_{(i+n,i+n)} \right] 
\end{align}
\end{subequations}
and, regarding the derivative with respect to the times $\tau_i$ ($i=1,\ldots, n$), we have
\begin{subequations}
\begin{gather}
\underset{\lambda = -4\tau_1}{\operatorname{res}} \operatorname{Tr}\left( \Gamma^{-1}\Gamma' \partial_{\tau_1} T_B \right)  = 4\nu \left(\frac{1}{2n}-1 \right) \left( \Gamma_{0}^{-1} \Gamma_{1} \right)_{(1.1)} + \nonumber \\
+ 4a^{(1)} \left(1- \frac{1}{2n} \right)  \left(-\Gamma_{0}^{-1} \Gamma_{1}\Gamma_{0}^{-1}\Gamma_{1} + 2\Gamma_{0}^{-1}\Gamma_{2}\right)_{(1,1)} \\
\underset{\lambda = -4\tau_i}{\operatorname{res}}\operatorname{Tr}\left( \Gamma^{-1}\Gamma' \partial_{\tau_i} T_B \right) 
= 4\nu \left(\frac{1}{n}-1 \right)\left[ \left( \Gamma_{0}^{-1} \Gamma_{1} \right)_{(i.i)}  + \left( \Gamma_{0}^{-1} \Gamma_{1} \right)_{(i+n,i+n)} \right] + \nonumber \\
+ 4a^{(i)} \left(1- \frac{1}{n} \right)\left[  \left(-\Gamma_{0}^{-1} \Gamma_{1} \Gamma_{0}^{-1}\Gamma_{1} + 2\Gamma_{0}^{-1}\Gamma_{2}\right)_{(i,i)} + \right. \nonumber \\
\left. + \left(-\Gamma_{0}^{-1} \Gamma_{1} \Gamma_{0}^{-1} \Gamma_{1} + 2\Gamma_{0}^{-1}\Gamma_{2}\right)_{(i+n,i+n)}  \right]
\end{gather}
\end{subequations}
where the $\Gamma_i$'s are coefficients of the asymptotic expansion of the matrix $\Gamma$ near $\infty$ and $-4\tau_j$. We recall that each asymptotic expansion (the $\Gamma_i$'s) is different for each point $-4\tau_j$ and it's different from the one near $\infty$.

The residue at infinity does not give any contribution in either case.
\end{thm}

\begin{proof}
We calculate the derivatives of the conjugation factor 
\begin{equation}
\renewcommand\arraystretch{1.7}
\begin{array}{ll}
\partial_{a^{(1)}} T_B(\lambda) &= \text{diag}\left[ \partial_{a^{(1)}}\left(\theta_1 - \frac{\kappa}{2n}\right), 0, \ldots,  0  \right] \\
&= \text{diag}\left[ \frac{1}{\lambda_1} \left( \frac{1}{2n}-1\right), 0, \ldots,  0  \right] = \frac{1}{\lambda_1} \left(\frac{1}{2n}-1\right) \cdot E_{(1,1)} \\
\partial_{a^{(i)}} T_B(\lambda) &= \text{diag}\left[0, \ldots, \partial_{a^{(i)}}\left(\theta_i - \frac{\kappa}{2n}\right), \ldots,  \partial_{a^{(i)}}\left(\theta_i - \frac{\kappa}{2n}\right), \ldots, 0  \right] \\
&= \frac{1}{\lambda_i} \left( \frac{1}{n}-1\right) \cdot E_{(i,i),\, (i+n, i+n)} \\
\partial_{\tau_1} T_B(\lambda) &= \text{diag}\left[ \partial_{\tau_1}\left(\theta_1 - \frac{\kappa}{2n}\right), 0, \ldots,  0  \right] \\
&= \text{diag}\left[ \left(\frac{4a^{(1)}}{\lambda_1^2} - \frac{4\nu}{\lambda_1} \right)  \left(1- \frac{1}{2n}\right), 0, \ldots,  0  \right] \\
&= \left( \frac{4a^{(1)}}{\lambda_1^2} - \frac{4\nu}{\lambda_1} \right)  \left(1- \frac{1}{2n}\right) \cdot E_{(1,1)} \\
\partial_{\tau_i} T_B(\lambda) &= \text{diag}\left[0, \ldots, \partial_{\tau_i}\left(\theta_i - \frac{\kappa}{2n}\right), \ldots,  \partial_{\tau_i}\left(\theta_i - \frac{\kappa}{2n}\right), \ldots, 0  \right] \\
&= \left( \frac{4a^{(i)}}{\lambda_i^2} - \frac{4\nu}{\lambda_i} \right)  \left(1- \frac{1}{n}\right) \cdot E_{(i,i),\, (i+n, i+n)}
\end{array}
\end{equation}
where $E_{(i,i)\, (i+n, i+n)}$ is the zero matrix with only two non-zero entries (which are $1$'s) in the $(i,i)$ and $(i+n, i+n)$ positions.

Then, recalling the (formal) asymptotic expansion of the matrix $\Gamma$ near $\infty$ and $-4\tau_i$ for all $i$ (see \cite{Wasow} for a detailed discussion on the topic), the results follow from straightforward calculations.

\end{proof}

\section*{Acknowledgements}

The author gratefully acknowledges Dr. Marco Bertola for proposing this problem and for his valuable  help in solving it.

\bibstyle{plain}
\bibliography{Bessel}

\begin{thebibliography}{10}

\bibitem{Misomonodromic}
M.~Bertola.
\newblock The dependence of the monodromy data of the isomonodromic tau
  function.
\newblock {\em Comm. Math. Phys.}, 294(2):539--579, 2010.

\bibitem{MeMmulti}
M.~Bertola and M.~Cafasso.
\newblock Riemann-{H}ilbert approach to multi-time processes: the {A}iry and
  the {P}earcey cases.
\newblock {\em Physica D}, 241(23-24):2237--2245, 2012.

\bibitem{MeM}
M.~Bertola and M.~Cafasso.
\newblock The transition between the {G}ap {P}robabilities from the {P}earcey
  to the {A}iry {P}rocess - a {R}iemann--{H}ilbert {A}pproach.
\newblock {\em Int. Math. Res. Not.}, 2012(7):1519--1568, 2012.

\bibitem{Dyson}
F.~J. Dyson.
\newblock A brownian motion model for the eigenvalues of a random matrix.
\newblock {\em J. Math. Phys.}, 3:1191 -- 1198, 1962.

\bibitem{EM}
B.~Eynard and M.~L. Mehta.
\newblock Matrices coupled in a chain. {I}. {E}igenvalue correlations.
\newblock {\em J. Phys. A: Math. Gen.}, 31:4449--4456, 1998.

\bibitem{Ptrans}
A.~Fokas, A.~Its, A.~Kapaev, and V.~Novokshenov.
\newblock {\em Painlev\'e Transcendents. The Riemann Hilbert Approach}, volume
  128 of {\em Mathematical Surveys and Monographs}.
\newblock American Mathematical Society, Providence, RI, 2006.

\bibitem{forresterLUE}
P.~J. Forrester.
\newblock The spectrum edge of random matrix ensembles.
\newblock {\em Nucl. Phys. B}, 402(3):709 -- 728, 1993.

\bibitem{forrester}
P.~J. Forrester, T.~Nagao, and G.~Honner.
\newblock Correlations for the orthogonal-unitary and symplectic-unitary
  transitions at the hard and soft edges.
\newblock {\em Nucl. Phys. B}, 553(3):601--643, 1999.

\bibitem{JohnSasha}
J.~Harnad and A.~R. Its.
\newblock Integrable {F}redholm operators and dual isomonodromic deformations.
\newblock {\em Comm. Math. Phys.}, 226(3):497--530, 2002.

\bibitem{IIKS}
A.~R. Its, A.~G. Izergin, V.~E. Korepin, and N.~A. Slavnov.
\newblock Differential equations for quantum correlation differential equations
  for quantum correlation functions.
\newblock {\em Int. J. Mod. Phys.}, B4:1003 -- 1037, 1990.

\bibitem{JMUII}
M.~Jimbo and T.~Miwa.
\newblock Monodromy preserving deformation of linear ordinary differential
  equations with rational coefficients. {II}.
\newblock {\em Phys. D}, 2(3):407--448, 1981.

\bibitem{JMUIII}
M.~Jimbo and T.~Miwa.
\newblock Monodromy preserving deformation of linear ordinary differential
  equations with rational coefficients. {III}.
\newblock {\em Phys. D}, 1:26--46, 1981/82.

\bibitem{JMU}
M.~Jimbo, T.~Miwa, and K.~Ueno.
\newblock Monodromy preserving deformation of linear ordinary differential
  equations with rational coefficients. i. general theory and $\tau$-function.
\newblock {\em Phys. D}, 2(2):306--352, 1981.

\bibitem{KarlinMcGregor}
S.~Karlin and J.~McGregor.
\newblock Coincidence probabilities.
\newblock {\em Pacific J. Math.}, 9(4):1141 -- 1164, 1959.

\bibitem{Japan}
M.~Katori and H.~Tanemura.
\newblock Non-colliding squared {B}essel processes.
\newblock {\em J. Stat. Phys.}, 142:592--615, 2011.

\bibitem{KuijVan}
A.~B.~J. Kuijlaars and M.~Vanlessen.
\newblock Universality for eigenvalue correlations from the modified jacobi
  unitary ensemble.
\newblock {\em Int. Math. Res. Not.}, 30:1575--1600, 2002.

\bibitem{Mehta}
M.~L. Mehta.
\newblock {\em Random Matrices}.
\newblock Elsevier/Academic Press, Amsterdam, third edition, 2004.

\bibitem{NagFor}
T.~Nagao and P.~J. Forrester.
\newblock Asymptotic correlations at the spectrum edge of random matrices.
\newblock {\em Nucl. Phys. B}, 435(3):401--420, 1995.

\bibitem{NagWad}
T.~Nagao and M.~Wadati.
\newblock Eigenvalue distribution of random matrices at the spectrum edge.
\newblock {\em J. Phys. Sec. Japan}, 62(11):3845--3856, 1993.

\bibitem{traceideals}
B.~Simon.
\newblock {\em Trace Ideals and their applications}, volume 120 of {\em
  Mathematical Surveys and Monographs}.
\newblock American Mathematical Society, Providence, RI, {S}econd edition,
  2005.

\bibitem{Soshnikov}
A.~Soshnikov.
\newblock Determinantal random point fields.
\newblock {\em Russian Mathematical Surveys}, 55:923 -- 975, 2000.

\bibitem{TWrandom}
C.~Tracy and H.~Widom.
\newblock {\em Introduction to Random Matrices}, volume 424 of {\em Lecture
  Notes in Physics}.
\newblock Sringer, Berlin, 1993.

\bibitem{TWFredh}
C.~Tracy and H.~Widom.
\newblock Fredholm determinants, differential equations and matrix models.
\newblock {\em Comm. Math. Phys.}, 163:33--72, 1994.

\bibitem{TWBessel}
C.~Tracy and H.~Widom.
\newblock Level spacing distributions and the {B}essel kernel.
\newblock {\em Comm. Math. Phys.}, 161(2):289 -- 309, 1994.

\bibitem{TWMBessel}
C.~Tracy and H.~Widom.
\newblock Differential equations for {D}yson processes.
\newblock {\em Comm. Math. Phys.}, 252(1-3):7--41, 2004.

\bibitem{justVan}
M.~Vanlessen.
\newblock Strong asymptotics of {L}aguerre-type orthogonal polynomials and
  applications in random matrix theory.
\newblock {\em Constr. Approx.}, 25(2):125--175, 2007.

\bibitem{Wasow}
W.~Wasow.
\newblock {\em Asymptotic expansions for Ordinary Differential Equations},
  volume XIV of {\em Pure and Applied Mathematics}.
\newblock Interscience Publishers, 1965.

\end{thebibliography}

\end{document}